\PassOptionsToPackage{unicode}{hyperref}
\PassOptionsToPackage{hyphens}{url}
\PassOptionsToPackage{dvipsnames,svgnames,x11names}{xcolor}
\documentclass[12pt]{article}

\usepackage{amsmath,amssymb,amsthm,bm,mathtools}
\DeclareUnicodeCharacter{2212}{\textminus}
\usepackage{iftex}
\ifPDFTeX
\usepackage[T1]{fontenc}
\usepackage[utf8]{inputenc}
\usepackage{enumitem}
\usepackage{textcomp} 
\else 
  \usepackage{unicode-math}
  \defaultfontfeatures{Scale=MatchLowercase}
  \defaultfontfeatures[\rmfamily]{Ligatures=TeX,Scale=1}
\fi
\usepackage{lmodern}
\ifPDFTeX\else  
\fi
\IfFileExists{upquote.sty}{\usepackage{upquote}}{}
\IfFileExists{microtype.sty}{
  \usepackage[]{microtype}
  \UseMicrotypeSet[protrusion]{basicmath} 
}{}
\makeatletter
\@ifundefined{KOMAClassName}{
  \IfFileExists{parskip.sty}{%
    \usepackage{parskip}
  }{
    \setlength{\parindent}{0pt}
    \setlength{\parskip}{6pt plus 2pt minus 1pt}}
}{
  \KOMAoptions{parskip=half}}
\makeatother
\usepackage{xcolor}
\makeatletter
\ifx\paragraph\undefined\else
  \let\oldparagraph\paragraph
  \renewcommand{\paragraph}{
    \@ifstar
      \xxxParagraphStar
      \xxxParagraphNoStar
  }
  \newcommand{\xxxParagraphStar}[1]{\oldparagraph*{#1}\mbox{}}
  \newcommand{\xxxParagraphNoStar}[1]{\oldparagraph{#1}\mbox{}}
\fi
\ifx\subparagraph\undefined\else
  \let\oldsubparagraph\subparagraph
  \renewcommand{\subparagraph}{
    \@ifstar
      \xxxSubParagraphStar
      \xxxSubParagraphNoStar
  }
  \newcommand{\xxxSubParagraphStar}[1]{\oldsubparagraph*{#1}\mbox{}}
  \newcommand{\xxxSubParagraphNoStar}[1]{\oldsubparagraph{#1}\mbox{}}
\fi
\makeatother

\usepackage{longtable,booktabs,array}
\usepackage{calc} 
\usepackage{etoolbox}
\makeatletter
\patchcmd\longtable{\par}{\if@noskipsec\mbox{}\fi\par}{}{}
\makeatother
\IfFileExists{footnotehyper.sty}{\usepackage{footnotehyper}}{\usepackage{footnote}}
\makesavenoteenv{longtable}
\usepackage{graphicx}
\makeatletter
\def\maxwidth{\ifdim\Gin@nat@width>\linewidth\linewidth\else\Gin@nat@width\fi}
\def\maxheight{\ifdim\Gin@nat@height>\textheight\textheight\else\Gin@nat@height\fi}
\makeatother
\setkeys{Gin}{width=\maxwidth,height=\maxheight,keepaspectratio}
\makeatletter
\def\fps@figure{htbp}
\makeatother

\makeatletter
\@ifpackageloaded{caption}{}{\usepackage{caption}}
\AtBeginDocument{%
\ifdefined\contentsname
  \renewcommand*\contentsname{Table of contents}
\else
  \newcommand\contentsname{Table of contents}
\fi
\ifdefined\listfigurename
  \renewcommand*\listfigurename{List of Figures}
\else
  \newcommand\listfigurename{List of Figures}
\fi
\ifdefined\listtablename
  \renewcommand*\listtablename{List of Tables}
\else
  \newcommand\listtablename{List of Tables}
\fi
\ifdefined\figurename
  \renewcommand*\figurename{Figure}
\else
  \newcommand\figurename{Figure}
\fi
\ifdefined\tablename
  \renewcommand*\tablename{Table}
\else
  \newcommand\tablename{Table}
\fi
}
\@ifpackageloaded{float}{}{\usepackage{float}}
\floatstyle{ruled}
\@ifundefined{c@chapter}{\newfloat{codelisting}{h}{lop}}{\newfloat{codelisting}{h}{lop}[chapter]}
\floatname{codelisting}{Listing}

\makeatother
\makeatletter
\@ifpackageloaded{caption}{}{\usepackage{caption}}
\@ifpackageloaded{subcaption}{}{\usepackage{subcaption}}
\makeatother

\ifLuaTeX
  \usepackage{selnolig}  
\fi
\usepackage[]{natbib}
\usepackage{bookmark}

\IfFileExists{xurl.sty}{\usepackage{xurl}}{} 
\hypersetup{
  pdftitle={Title},
  pdfauthor={Author 1; Author 2},
  pdfkeywords={3 to 6 keywords, that do not appear in the title},
  colorlinks=true,
  linkcolor={blue},
  filecolor={Maroon},
  citecolor={Blue},
  urlcolor={Blue},
  pdfcreator={LaTeX via pandoc}}

\newcommand{\anon}{1}
\theoremstyle{plain}

\newtheorem{theorem}{Theorem}
\newtheorem{proposition}{Proposition}
\newtheorem{lemma}{Lemma}
\newtheorem{definition}{Definition}

\newtheorem{assumption}{Assumption}

\begin{document}

\def\spacingset#1{\renewcommand{\baselinestretch}%
{#1}\small\normalsize} \spacingset{1}


\if1\anon
{
  \title{\bf Multiple Testing of Partial Conjunction Hypotheses for Assessing Replicability Across Dependent Studies
}
  \author{Monitirtha Dey\thanks{Institute for Statistics, University of Bremen. Dey gratefully acknowledges financial support from the German Research Foundation (DFG) via Grant No. DI 1723/5-3.},~Trambak Banerjee\thanks{Analytics, Information and Operations, University of Kansas. Corresponding author: trambak@ku.edu},~Prajamitra Bhuyan\thanks{Operations Management, Indian Institute of Management, Calcutta.},\hspace{3cm}\\and Arunabha Majumdar\thanks{Department of Mathematics, Indian Institute of Technology, Hyderabad.}}
  \maketitle
} \fi

\if0\anon
{
  \bigskip
  \bigskip
  \bigskip
  \begin{center}
    {\LARGE\bf  Assessing Replicability Across Dependent Studies: A Framework for Testing Partial Conjunction Hypotheses with Application to GWAS}
\end{center}
  \medskip
} \fi
\vspace{-5pt}
\bigskip
\begin{abstract}
Replicability is central to scientific progress, and the partial conjunction (PC) hypothesis testing framework provides an objective tool to quantify it across disciplines. Existing PC methods assume independent studies. Yet many modern applications, such as genome-wide association studies (GWAS) with sample overlap, violate this assumption, leading to dependence among study-specific summary statistics. Failure to account for this dependence can drastically inflate type I errors when combining inferences. We propose e-Filter, a powerful procedure grounded on the theory of e-values. It involves a filtering step that retains a set of the most promising PC hypotheses, and a selection step where PC hypotheses from the filtering step are marked as discoveries whenever their e-values exceed a selection threshold. We establish the validity of e-Filter for FWER and FDR control under unknown study dependence. A comprehensive simulation study demonstrates its excellent power gains over competing methods.  We apply e-Filter to a GWAS replicability study to identify consistent genetic signals for low-density lipoprotein cholesterol (LDL-C). Here, the participating studies exhibit varying levels of sample overlap, rendering existing methods unsuitable for combining inferences. A subsequent pathway enrichment analysis shows that e-Filter replicated signals achieve stronger statistical enrichment on biologically relevant LDL-C pathways than competing approaches.  
\end{abstract}

\noindent%
\vfill

\newpage
\spacingset{1.5} 

\section{Introduction}
\label{sec:intro}
Replicability is fundamental to scientific progress, as consistent results across multiple studies enhance the credibility of new discoveries. 
The partial conjunction hypothesis (PCH) testing framework \citep{Friston2005,benjamini2008screening,benjamini2009selective} provides an objective inferential tool for determining whether a scientific finding is replicated across several studies. While PCH testing finds applications across myriad domains, such as neuroscience \citep{benjamini2008screening,Goeman2018}, genetics \citep{SunWei2011}, causal inference \citep{Karmakar}, the existing literature routinely assumes that the study-specific summary statistics are independent. See, for instance, \cite{Owen,Dickhaus2024,tran2025covariateadaptivetestreplicabilitymultiple,liang2025powerful,lyu2025testingcompositenullhypotheses} and a recent review by \cite{bogomolov2023replicability}. However, many modern applications violate this assumption, leading to dependence among
study-specific summary statistics. 
A typical example, for instance,  arises in replicability analysis of genome-wide association studies (GWAS), where millions of single-nucleotide polymorphisms (SNPs) are tested to identify variants that show consistent associations with a disease or trait across studies. Unlike meta-analysis, which tests whether a genetic variant–trait association exists in at least one study, replication requires that multiple studies detect the same association, yielding a more reliable set of genuine genetic signals. In GWAS, however, the effects of susceptibility SNPs on complex traits are typically small and require large sample sizes for reliable detection. To achieve this, contemporary studies often incorporate data from large biobanks, such as UK Biobank, BioVU, and 23andMe, which can lead to overlapping participants across studies. 
	Moreover, shared controls are frequently used for multiple studies on a complex disease to reduce genotyping or sequencing costs \citep{Han2016}. Consequently, publicly available summary statistics, such as $p-$values {of genetic association}, from such studies may exhibit dependence, as {it is challenging to share individual-level genetic data due to privacy constraints}. This aspect is particularly important, since ignoring study dependence in replicability analyses can inflate type I error {rates} when combining inferences. Figure \ref{fig:motfig1}, discussed later, demonstrates this phenomenon with the help of a simple example. 
	
	In this paper, we propose $e-$Filter, a powerful framework for identifying replicating signals across multiple dependent studies. To assess replicability, $e-$Filter relies on the PCH testing framework and is based on the theory of $e-$values. It involves two main steps: a filtering step that retains a set of most promising PC hypotheses, and a selection step where PC hypotheses from the filtering step are marked as discoveries whenever their $e-$values exceed a data-driven selection threshold. We prove that $e-$Filter provides valid familywise error rate (FWER) and false discovery rate (FDR) control under unknown dependence between the study-specific $p-$values. Our work is motivated by an unique GWAS replicability study aimed at identifying consistent genetic signals associated with low-density lipoprotein cholesterol (LDL-C), a major risk factor for cardiovascular diseases. Section \ref{sec:realdata} provides a detailed account of this application, where the participating studies exhibit varying levels of sample overlap, rendering existing methods unsuitable for simultaneous error-rate control. We apply $e-$Filter to this replicability study and a subsequent pathway enrichment analysis reveals that the genetic variants replicated by $e-$Filter prioritize biologically relevant pathways to LDL-C, yielding a statistically stronger enrichment than competing approaches.
	An alternative approach to replicability analysis in this setting is to construct PC $p-$values that are valid under arbitrary dependence between studies, such as the Bonferroni PC $p-$value \citep{benjamini2008screening}. However, testing procedures based on these $p$-values can be overly conservative (see Section \ref{sec:background}). In contrast, $e-$Filter provides a more powerful alternative in such scenarios. The numerical experiments in Section \ref{sec:sims} demonstrate that $e-$Filter offers substantial power gains over standard approaches that rely on such PC $p-$values.

	\begin{figure}[!h]
		\centering
		\includegraphics[width=0.8\linewidth]{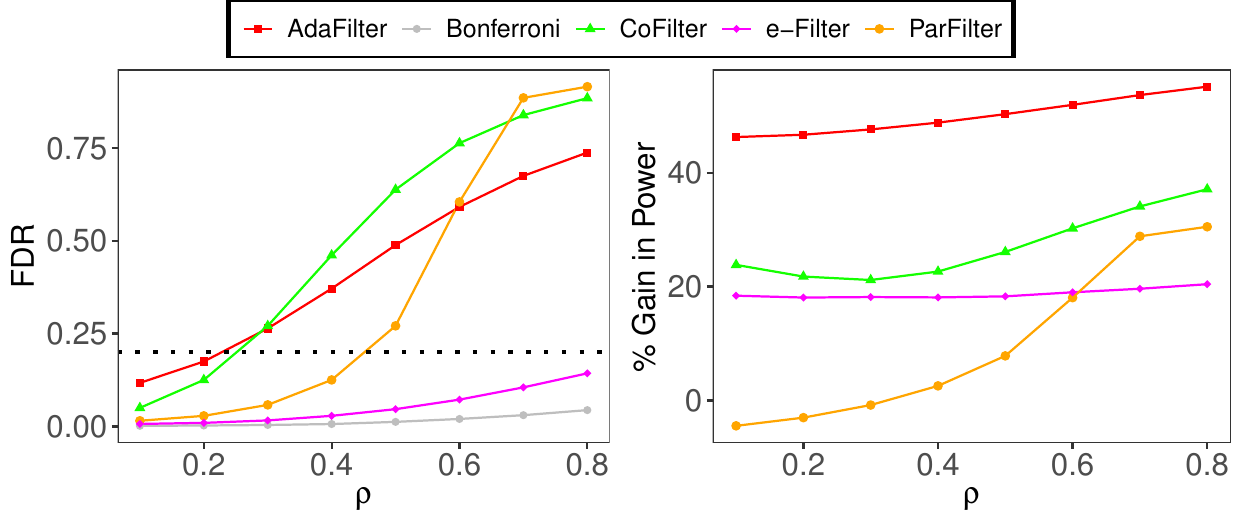}
		\caption{Consider five studies, each testing the association of 
			$m=10,000$ SNPs with a disease of interest. The objective is to test, for each $j=1,\ldots,m$, the PC null hypothesis that SNP 
			$j$ is significantly associated with the disease in at most one of the five studies. The dependence between the studies is governed by $\rho$ which is the coefficient of linear correlation between the five test statistics underlying each PC null. See Section \ref{app: fig_details} for the data-generation scheme. Left: Average false discovery rate (FDR). Right: Average \% gain in power over the Bonferroni PC method.}
		\label{fig:motfig1}
	\end{figure}
	
	Figure \ref{fig:motfig1} illustrates the consequences of ignoring study dependence in a toy GWAS replicability study. Consider five studies, each testing the association of 
	$m=10,000$ SNPs with a disease of interest. The objective is to test, for each $j=1,\ldots,m$, the PC null hypothesis that SNP 
	$j$ is significantly associated with the disease in at most one of the five studies. Thus, rejecting the $j^{th}$ PC null implies that the association of the $j^{th}$ SNP with the disease is replicated in at least two of the five studies. The dependence between the studies is governed by the parameter $\rho$ which is the coefficient of linear correlation between the five test statistics underlying each PC null hypothesis. For a range of values of $\rho$, the left panel of Figure \ref{fig:motfig1} reports the average false discovery rate (FDR) of the following five PC testing procedures across $500$ Monte-Carlo repetitions of the data generating scheme: the Bonferroni PC method, $e-$Filter, AdaFilter \citep{Owen}, CoFilter \citep{Dickhaus2024} and ParFilter \citep{tran2025covariateadaptivetestreplicabilitymultiple}. Here, the last three methods ignore study dependence. The results show that, unlike Bonferroni and $e$-Filter, AdaFilter, CoFilter, and ParFilter fail to control the FDR at the nominal level $\alpha=0.2$ for moderate to large $\rho$, due to their underlying assumption of study independence. The right panel reports the average percentage gain in power over the Bonferroni method and reveals that $e-$Filter is the only approach achieving nearly a 20\% power gain while maintaining FDR control across all $\rho$ values. We discuss $e-$Filter in Section \ref{sec:e_filter} while Section \ref{app: fig_details} in the supplement provides the data-generation scheme underlying Figure \ref{fig:motfig1}.
	All R codes for reproducing the numerical illustrations in this paper are available at \url{https://github.com/trambakbanerjee/efilter}. 
	\section{Multiple testing of partial conjunction hypotheses}
	\label{sec:background}
	%
	\subsection{Problem setup}
	\label{sec:setup}
	Consider a setting where $n$ studies are each testing $m$ base null hypotheses. Denote the set of base null hypotheses by $\mathcal H_0=\{H_{0ij}:i=1,\ldots,n\text{ and }j=1,\ldots,m\}$ and let $\mathcal P=(P_{ij}:i=1,\ldots,n\text{ and }j=1,\ldots,m)$ denote the $n\times m$ matrix of base $p-$values for $\mathcal H_0$. Throughout this paper, we assume that each $P_{ij}$ is a valid $p-$value in the sense that $\mathrm{pr} (P_{ij}\le \gamma)\le \gamma$ under $H_{0ij}$. However, we do not make any assumptions on the dependence between $P_{ij}$ and $P_{i'j}$, i.e, in studies $i$ and $i'$ the $p-$values for the $j^{th}$ base null may not necessarily be independent. Denote $\mathcal H_{0j}=\{H_{0ij}:i=1,\ldots,n\}$ and let $k_j$ be the  unknown number of false null hypotheses in $\mathcal H_{0j}$. For $2\le r\le n$ we are interested in the following multiple testing problem: 
	\begin{equation}
		\label{eq:multipletesting}
		H_{0j}^{r/n}: k_j<r~\text{versus}~H_{1j}^{r/n}: k_j\ge r,~~j=1,\ldots,m.
	\end{equation}
	When $r=1$, Equation \eqref{eq:multipletesting} represents the global null for meta-analysis and rejecting it suggests that the $j^{th}$ base null hypothesis is false in at least one of the $n$ studies. In contrast, rejecting $H_{0j}^{r/n}$ for $2\le r\le n$ implies that the finding is replicated in at least $r$ of the $n$ studies.
	Let $\theta_j = \mathbb I(H_{0j}^{r/n}~\text{is false})$ denote an indicator function that gives the true state of the $j^{th}$ testing problem and denote $\delta_j \in\{0, 1\}$ as the decision that we make about $H_{0j}^{r/n}$ with $\delta_j = 1$ being a decision to reject $H_{0j}^{r/n}$.
	A selection error, or false positive, occurs if $\delta_j = 1$ while $\theta_j=0$. In multiple testing problems such false discoveries, denoted $V=\sum_{j=1}^{m}\delta_j(1-\theta_j)$, are inevitable and in this paper we consider the following three standard notions of controlling false discoveries:
	$\text{FWER}\coloneqq\mathrm{pr}(V\ge 1),~~\text{PFER}\coloneq E(V),~~\text{FDR}\coloneq E\{V/\max(R, 1)\},
	$
	where $R = \sum_{j=1}^{m}\delta_j$ is the total number of discoveries. 
	\subsection{Conventional $p-$value procedures}
	\label{sec:pval_approach}
	Denote $\bm{P}_j=(P_{ij}:i=1,\ldots,n)$ and suppose $P_{(1)j} \leq P_{(2)j} \leq \cdots \leq P_{(n)j}$ are the ordered $p$-values for $\mathcal H_{0j}$. Conventional $p-$value procedures for Problem \eqref{eq:multipletesting} can be broadly classified into two groups: the first group relies on the construction of a valid PC $p-$value for testing $H_{0j}^{r/n}$ and commonly includes the following four methods:
	(1) Bonferroni's method: $P_{r / n,j}^B=(n-r+1) P_{(r)j}$, (2) Simes' method: $P_{r / n,j}^S=\min _{r \leq i \leq n}[\{(n-r+1)/(i-r+1)\} P_{(i)j}]$, (3) Fisher's method: $P_{r / n,j}^F=\mathrm{pr}(\chi_{2(n-r+1)}^2 \geq-2 \sum_{i=r}^n \log P_{(i)j})$, and (4) Cauchy combination method: $P_{r / n,j}^C=\mathrm{pr}\{W\geq(n-r+1)^{-1}\sum_{i=r}^n \tan\{(0.5-P_{(i)j})\pi\}\}$, where $W$ denotes a standard Cauchy random variable.
	The first three methods are discussed in \cite{benjamini2008screening}, while the fourth one is based on \cite{liu2020cauchy} on combining dependent $p-$values. Lemma \ref{lem:cauchy_pvals} establishes that under some conditions $P_{r / n,j}^C$ is an approximately valid $p-$value under $H_{0j}^{r / n}$ for arbitrary dependence among the test statistics underlying $\bm P_j$. 
	\begin{lemma}
		\label{lem:cauchy_pvals}
		Let $\bm X_j$ denote the test statistics corresponding to $\bm P_j$. If $E(\bm X_j)=\bm 0$ and for any $1\le r < s\le n$, $(X_{rj}, X_{sj})$ follows a bivariate normal distribution, then $P_{r / n,j}^C$ is an approximately valid $p-$value under $H_{0j}^{r / n}$ for arbitrary dependence among the entries of $\bm X_j$. 
	\end{lemma}
	Notably, amongst these four methods only $P^{B}_{r / n,j}$ and $P^{C}_{r / n,j}$ are valid $p-$values for $H_{0j}^{r/n}$ under arbitrary dependence among the entries of $\bm P_j$. In contrast, Simes' $p-$value ($P^{S}_{r / n,j}$) is valid, for instance, when $\bm P_j$ satisfy the positive regression dependency on a subset (PRDS) property (see \cite{benjamini2008screening}) and Fisher's $p-$value ($P^{F}_{r / n,j}$) is valid when the entries of $\bm P_j$ are independent. 
	
	Denote $P_{r / n,j}$ as a generic PC $p-$value for $H_{0j}^{r/n}$ and let $\mathcal H^{r/n}=\{H_{0j}^{r/n}:j=1,\ldots,m\}$. Given $m$ PC $p-$values $\mathcal P_{r/n}=\{P_{r / n,j}:j=1,\ldots,m\}$, a standard technique for testing $\mathcal H^{r/n}$ involves applying popular multiple testing adjustments, such as the BH procedure \citep{benjamini1995controlling}, to $\mathcal P_{r/n}$ for controlling the false discoveries at the desired level. However, this approach can be extremely conservative since the inequality $\mathrm{pr}(P_{r / n,j} \leq \gamma) \leq \gamma$ for a given $\gamma$ is only tight for the least favorable null, which is the setting where exactly $r-1$ base hypothesis are non-null. In contrast, when the global null is true, which is the case in most genetic applications, $P_{r/n,j}$ can be highly superuniform, i.e, $\mathrm{pr}(P_{r/n,j}\le \gamma)\ll\gamma$ \citep{liang2025powerful}. 
	
	This motivates procedures in the second group that employ various strategies to improve upon the power of standard PC $p$-value–based approaches \citep{Owen,Dickhaus2024,liang2025powerful,tran2025covariateadaptivetestreplicabilitymultiple,lyu2025testingcompositenullhypotheses}. For instance, \cite{Owen,Dickhaus2024} develop methods that typically involve an initial filtering step to filter out PC hypotheses in $\mathcal H^{r/n}$ that are most likely to be null. In applications where the global null is overwhelmingly dominant, this filtering step substantially improves the power of the underlying multiple testing procedure for testing $\mathcal H^{r / n}$. 
	However, these procedures are developed under the assumption that for each $j$ the entries of $\bm P_j$ are independent. In the following subsection we present a simple $e-$value based procedure that guarantees valid FDR control even if the $P_{ij}$'s are arbitrarily correlated for each $j$.
	\subsection{$e-$values in hypothesis testing and the $e-$PCH procedure}
	\label{sec:e-PCH}
	In the context of hypothesis testing, an $e-$value is a nonnegative random variable whose expected value is at most one under the null hypothesis. We will use the notation ``$e$'' to denote both the random variable and its realized value. While the classical theory of inference is built around the $p-$values, inference using $e-$values benefit from their natural
	connections to game-theoretic probability and statistics, flexibility and robustness in multiple testing under
	dependence, and their central role in anytime-valid statistical inference. See \cite{STA-002} for a comprehensive overview. One may obtain an $e-$value from a valid $p-$value $P$ using calibrators. A calibrator
	$\varphi : [0,1] \to [0,\infty)$ is an upper semicontinuous function that satisfies $\int_{0}^{1} \varphi(x)dx = 1$, which
	implies that $\varphi(P)$ is an $e-$value for any $p$-value $P$. 
	
	Next, consider the following three step procedure for testing $\mathcal H^{r/n}$: \textbf{Step 1: }for each $j$, convert the $p-$values in $\bm P_j$ to $e-$values $\bm e_j=\{e_{ij}=\varphi(P_{ij}):i=1,\ldots,n\}$ where $\varphi$ is a calibrator, such as those discussed in \cite{vovk2021values}; \textbf{Step 2: }define the PCH $e-$value for $H_{0j}^{r / n}$ as $e_j=(n-r+1)^{-1}\sum_{i=1}^{n-r+1}e_{(i)j}$, where $e_{(1)j}\le \cdots\le e_{(n)j}$; \textbf{Step 3: }apply the $e$-BH procedure of \cite{WangRamdas} on $e_1,\ldots,e_m$ for FDR control at the desired level, say, $\alpha$.
	
	We call this method the $e-$PCH procedure. Lemma \ref{lem:e-PCH} establishes that this procedure controls FDR at a pre-specified level $\alpha$ under arbitrary dependence among the base $p-$values $\mathcal P$.
	\begin{lemma}\label{lem:e-PCH}
		The $e-$PCH procedure controls FDR at the pre-specified level $\alpha$ under arbitrary dependence among the base $p-$values $\mathcal P$. 
	\end{lemma}
	Although valid, the $e$-PCH procedure can be highly conservative and 
	this conservatism arises because, under $H_{0j}^{r/n}$, the bound $E(e_j)\le1$ is tight only for the least favorable null. In contrast, under the global null $E(e_j)$ can be much smaller than 1. 
	\section{The $e-$Filter procedure}
	\label{sec:e_filter}
	\subsection{Motivation}
	\label{sec:efilter_motivation}
	Standard methods that rely directly on $P_{r/n,j}$ can be overly conservative for testing $H_{0j}^{r/n}$, since under the global null $P_{r/n,j}$ are highly superuniform. To overcome this conservativeness, procedures such as AdaFilter and CoFilter leverage an initial filtering step that retains only PC hypotheses whose PC $p-$values are smaller than a, carefully chosen, data-driven threshold $\gamma_0$. When the global null dominates all other base null configurations, the PC $p-$values are conservative and so few fall below $\gamma_0$. Restricting inference to this smaller subset of PC hypotheses substantially reduces the multiplicity of the underlying problem and boosts power. For example, AdaFilter uses Bonferroni PC $p-$values $P^B_{r/n,j}$ and defines 
	$F_j:=(n-r+1) P_{(r-1) j}$ for the filtering step. For FWER control at level $\alpha$, AdaFilter rejects $H_{0 j}^{r / n}$ if $P^B_{r/n,j}<\gamma_0$ where $
	\gamma_0=\sup \{\gamma \in[0, \alpha] \mid \gamma \sum_{j=1}^M \mathbb I(F_j<\gamma) \leq \alpha\}$.
	\cite{Owen} show that for finite $m$ this procedure controls FWER and PFER under independence of the $nm$ base $p-$values. 
	However, the replicability study described in Section \ref{sec:realdata} includes five studies that have overlapping samples. This induces dependence among their base $p$-values, thus invalidating procedures such as AdaFilter in our setting. 
	This motivates us to develop $e-$Filter which retains Type I error control under unknown dependence between the study-specific $p-$values (Theorems \ref{thm:e_adafilter_fwer_control} and \ref{thm:asymp_e_filter_fdr_control}). 
	
	The main idea behind the $e-$Filter procedure is to transform $P^B_{r/n,j}$ to an $e-$value using a calibrator $\varphi$. Thereafter, for FWER control at level $\alpha$, we reject $H_{0j}^{r/n}$ if $\varphi(P^B_{r/n,j})>1/\gamma_e$ where $\gamma_e$ is the largest value of $\gamma\in[0,\alpha]$ such that $\gamma\sum_{j=1}^{m}\mathbb{I}\{\varphi(F_j)>1/\gamma\}\le \alpha$. 
	When the base $p-$values across studies are dependent, this, apparently simple, ``tweak'' to the AdaFilter procedure allows $e-$Filter to control the false discoveries even under study dependence. 
	\begin{figure}[!h]
		\centering
		\includegraphics[width=0.9\linewidth]{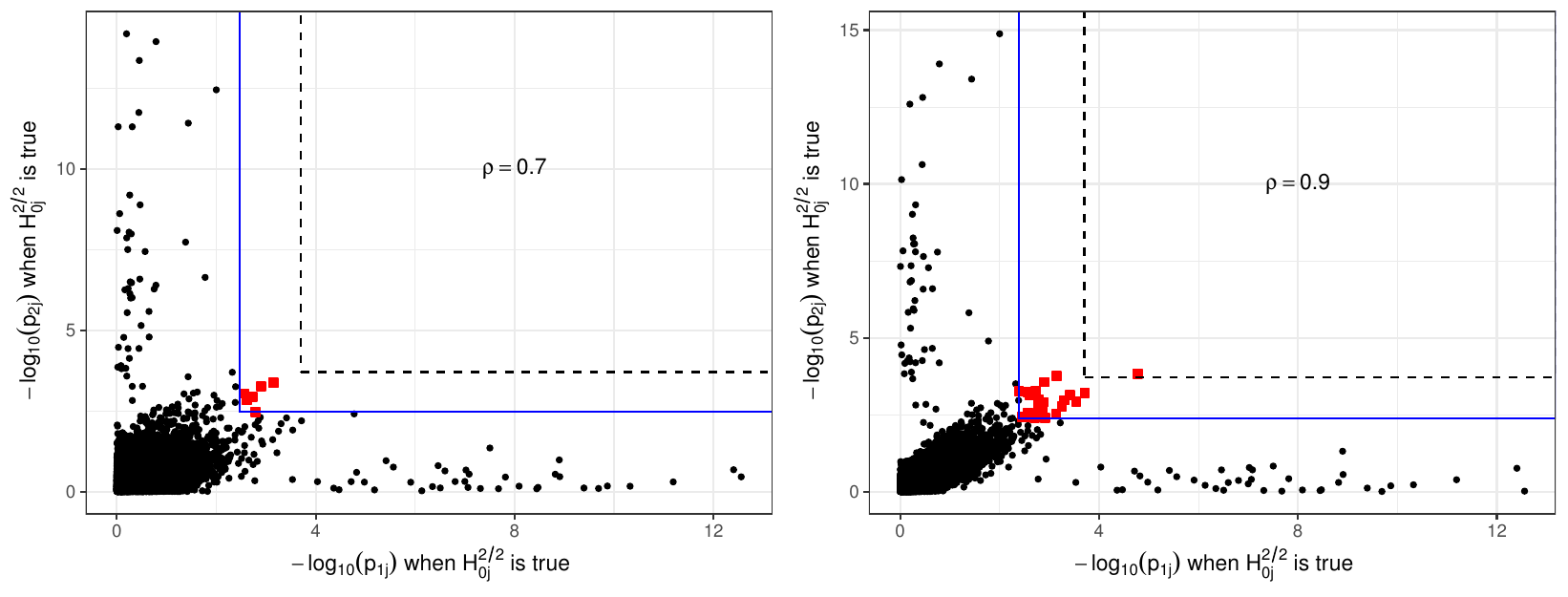}
		\caption{We test $\mathcal H^{2/2}$ with $m=10,000$. The base $p-$values are given by $P_{ij}=2\Phi(-|X_{ij}|)$, where the correlation between $X_{ij}$ and $X_{i'j}$ is $\rho$. Details of the data generating scheme are provided in Section \ref{app: fig_details}. The goal is to control the PFER at $\alpha=1$. The black dots are $-\log_{10}(P_{ij})$ whenever $H_{0j}^{2/2}$ is true. The solid vertical and horizontal lines mark $-\log_{10}(\gamma_0)$ so that AdaFilter rejects $H_{0j}^{2/2}$ whenever $-\log_{10}\{\max(P_{1j},P_{2j})\}>-\log_{10}(\gamma_0)$. The smaller region inside the dotted vertical and horizontal lines represents the rejection region of $e-$Filter. The rejections  are highlighted in red. Left panel: $\rho=0.7$. Right panel: $\rho=0.9$.}
		\label{fig:fig_sec3_1}
	\end{figure}
	We illustrate this in Figure \ref{fig:fig_sec3_1} with the help of an example while formal theories are developed in Section \ref{sec:theory}.
	\begin{figure}[!h]
		\centering
		\includegraphics[width=0.9\linewidth]{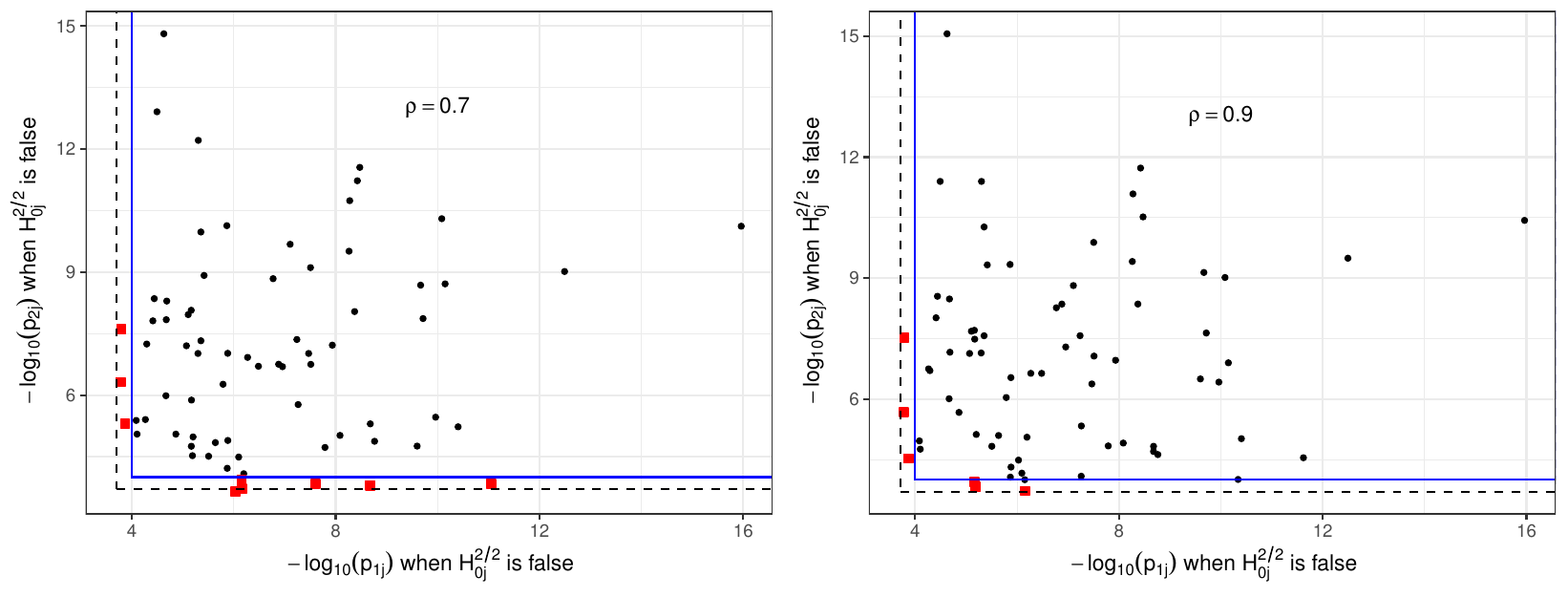}
		\caption{The black dots are $-\log_{10}(P_{ij})$ whenever the PC nulls are false. Here the region inside the solid vertical and horizontal lines represent the rejection region of the Bonferroni procedure that rejects $H_{0j}^{2/2}$ if $-\log_{10}\{\max(P_{1j},P_{2j})\}>-\log_{10}(\alpha/m)$. The rejection region of $e-$Filter is the larger region within the vertical and horizontal dotted lines. The red dots are PC nulls rejected by $e-$Filter but not by the Bonferroni procedure. Left panel: $\rho=0.7$. Right panel: $\rho=0.9$.}
		\label{fig:fig_sec3_2}
	\end{figure}
	
	We test $\mathcal H^{2/2}$ with $m=10,000$. The base $p-$values are given by $P_{ij}=2\Phi(-|X_{ij}|)$, where the correlation between $X_{ij}$ and $X_{i'j}$ is $\rho=0.7$. Details of the data generating scheme are provided in Section \ref{app: fig_details} of the supplement. The goal is to control the PFER at $\alpha=1$. The left panel of Figure \ref{fig:fig_sec3_1} plots $-\log_{10}(P_{ij})$ whenever $H_{0j}^{2/2}$ is true. Here the solid vertical and horizontal lines mark $-\log_{10}(\gamma_0)$ so that AdaFilter rejects $H_{0j}^{2/2}$ whenever $-\log_{10}\{\max(P_{1j},P_{2j})\}>-\log_{10}(\gamma_0)$. In this setting, AdaFilter falsely rejects $6~(>\alpha)$ PC nulls, which are highlighted in red. In contrast, $e-$Filter employs a relatively more stringent rejection criteria, represented by the smaller region inside the dotted vertical and horizontal lines, and does not make any false discoveries. In the right panel $\rho=0.9$ and while AdaFilter makes several more false discoveries, $e-$Filter makes only one. Figure \ref{fig:fig_sec3_1} may give the impression that the stricter rejection criteria of $e-$Filter may render it too conservative. Figure \ref{fig:fig_sec3_2} reveals that this is not necessarily the case. The left panel of Figure \ref{fig:fig_sec3_2} plots $-\log_{10}(P_{ij})$ whenever the PC nulls are false. Here the region inside the solid vertical and horizontal lines represent the rejection region of the Bonferroni procedure that rejects $H_{0j}^{2/2}$ if $-\log_{10}\{\max(P_{1j},P_{2j})\}>-\log_{10}(\alpha/m)$. The rejection region of $e-$Filter is the larger region within the vertical and horizontal dotted lines. The red dots are PC nulls rejected by $e-$Filter but not by the Bonferroni procedure, indicating the higher power of $e-$Filter in this example. The right panel considers $\rho=0.9$ and leads to the same conclusion. The numerical experiments in Section \ref{sec:sims} reveal that for some settings $e-$Filter can be substantially more powerful than the Bonferroni procedure.
	
	While the discussion in this section posits $e-$Filter as an ``$e-$value analogue'' of AdaFilter, we remark that unlike AdaFilter, $e-$Filter does not require independent studies and  it is applicable much more broadly, as long as any valid $p-$value for testing the global null under arbitrary dependence, such as \cite{wilson2019harmonic,liu2020cauchy,gui2025aggregating}, is available. 
	We formalize this comment in the next section.
	\subsection{Formal definitions}
	\label{sec:efilter_formal}
	Let $k=n-r+1$. For the $j^{th}$ PC null, define the PC $p-$value as
	\begin{equation}
		\label{eq:sj}
		S_j\coloneqq S(\bm P_j) = f_0(P_{(r)j},\ldots,P_{(n)j}),  
	\end{equation}
	where $f_0:[0,1]^{k}\to[0,1]$ is a generic $p-$value combination function that satisfies the following: (i) $f_0$ is non-decreasing in each argument, (ii) $f_0$ is symmetric, and (iii) $f_0(P_{1j},\ldots,P_{kj})$ is a valid $p-$value for testing the global null under arbitrary dependence between the base $p-$values $(P_{1j},\ldots,P_{kj})$. Examples of such combination function includes the Bonferroni $p-$value  $f_0(P_{1j},\ldots,P_{kj})=k\min(P_{1j},\ldots,P_{kj})$ or the Cauchy combination $p-$value \citep{liu2020cauchy} $f_0(P_{1j},\ldots,P_{kj})=\mathrm{pr}(W\geq k^{-1}\sum_{i=1}^k \tan\{(0.5-P_{ij})\pi\})$, where $W$ denotes a standard Cauchy random variable. In Equation \eqref{eq:sj}, $S_j$ is the Benjamini-Heller partial conjunction $p-$value \citep{benjamini2008screening,wang2019admissibility} and by Lemma 1 of \cite{benjamini2008screening} it is a valid $p-$value for testing $H_{0j}^{r/n}$ under arbitrary dependence between the base $p-$values $\bm P_j$.
	%
	
	Corresponding to each $S_j$ define
	\begin{equation}
		\label{eq:fj}
		F_j\coloneq F(\bm P_j)=f_1(P_{(1)j},\ldots,P_{(r-1)j}),
	\end{equation}
	where $f_1:[0,1]^{r-1}\to[0,1]$ is such that {(i) $f_1$ is symmetric}, and (ii) $F_j\le S_j$. For instance, in the context of AdaFilter $F_j=(n-r+1)P_{(r-1)j}$. Furthermore, if $S_j=P_{r/n,j}^C$ then a choice of $F_j$ is $F_j= \tan\{(0.5-P_{(r-1)j})\pi\}$. Crucially, in our construction, while $S_j$ is a valid $p-$value under $H_{0j}^{r/n}$, $F_j$ in Equation \eqref{eq:fj} is not. 
	Definition \ref{def:ef} formally presents the $e-$Filter procedures for PFER and FDR control, respectively.
	\begin{definition}
		\label{def:ef}
		Suppose $\varphi$ is a calibrator. For a pre-specified level $\alpha$, we have the following two $e-$Filter procedures:
		\begin{enumerate}
			\renewcommand{\theenumi}{1.\arabic{enumi}}
			\item \label{def:eadaBon} ($e-$Filter PFER).  Reject $H_{0 j}^{r / n}$ if $\varphi(S_j)>1/\gamma_e^{\mathrm{PFER}}$ where
			$$
			\gamma_e^{\mathrm{PFER}}=\sup \Big[\gamma \in[0, \alpha]:\gamma \sum_{j=1}^m \mathbb I\{\varphi(F_j)>1/\gamma\} \leq \alpha\Big].
			$$
			\item \label{def:eadaBH}($e-$Filter FDR). Reject $H_{0 j}^{r / n}$ if $\varphi(S_j)>1/\gamma_e^{\mathrm{FDR}}$ where
			$$
			\gamma_e^{\mathrm{FDR}}=\sup \Big[\gamma \in[0, \alpha] : \frac{\gamma \sum_{j=1}^m \mathbb I\{\varphi(F_j)>1/\gamma\}}{\max(\sum_{j=1}^m \mathbb I\{\varphi(S_j)>1/\gamma\}, 1)} \leq \alpha\Big].
			$$
		\end{enumerate}
	\end{definition}
	Note that in Definition \ref{def:ef} $\varphi(F_j)$ is not a valid $e-$value under $H_{0j}^{r/n}$. In contrast, $\varphi(S_j)$ is a valid $e-$value under the PC null. 
	Next, we introduce the $e-$Filter adjusted 
	``$e-$values'', which are easier to compute yet yield rejection sets equivalent to those from Definition \ref{def:ef}.
	\begin{definition}
		\label{def:ef_evals}
		Denote $S_j^{\sf e}\coloneq\varphi(S_j)$ and suppose $S_{(1)}^{\sf e}\ge \ldots\ge S_{(m)}^{\sf e}$ denote the ordered $e-$values for $\mathcal H^{r/n}$. Define
		$m_{(j)}=\sum_{h=1}^{m}\mathbb I\{\varphi(F_h)\ge S_{(j)}^e\}
		$ and suppose $\alpha$ is a pre-specified level.
		\begin{enumerate}
			\renewcommand{\theenumi}{2.\arabic{enumi}}
			\item \label{def:ef_eval_Bon} The $e-$Filter PFER adjusted ``$e-$value'' for $H_{0(j)}^{r/n}$ is 
			$e_{(j)}^{\sf PFER}={S_{(j)}^{\sf e}}/{m_{(j)}}$. Reject $H_{0 j}^{r / n}$ if $e_{j}^{\sf PFER}>1/\alpha$. 
			\item \label{def:ef_eval_BH}The $e-$Filter FDR adjusted ``$e-$value'' for $H_{0(j)}^{r/n}$ is 
			$e_{(j)}^{\sf FDR}=\max_{h\ge j}\{h{S_{(j)}^{\sf e}}/{m_{(j)}}\}.$ Reject $H_{0 j}^{r / n}$ if $e_{j}^{\sf FDR}>1/\alpha$.
		\end{enumerate}
	\end{definition}
	Lemma \ref{lem:ef_equivalence} establishes that Definitions \ref{def:ef} and \ref{def:ef_evals} produce the same rejection sets.
	\begin{lemma}
		\label{lem:ef_equivalence}
		For a pre-specified level $\alpha$, the rejection sets $\{j:S_j^{\sf e}>1/\gamma_e^{\sf PFER}\}$ and $\{j:e_{j}^{\sf PFER}>1/\alpha\}$ are equivalent. Also, the rejection sets $\{j:S_j^{\sf e}>1/\gamma_e^{\sf FDR}\}$ and $\{j:e_{j}^{\sf FDR}>1/\alpha\}$ are equivalent.
	\end{lemma}
	We end this section with a remark on the choice of the calibrator $\varphi$. In our empirical analyses, we take $\varphi(x)=\kappa x^{\kappa-1},~\kappa\in (0,1)$, which appears in Equation (1) of \cite{vovk2021values}. The choice of this calibrator is motivated by its strong empirical performance in our simulations (Section \ref{sec:sims}) and its analytical tractability, which substantially facilitates our theoretical analyses. In Section \ref{sec:sims} we discuss a simple strategy for tuning $\kappa$.
	\section{Theory}
	\label{sec:theory}
	\subsection{An improved upper bound on $\mathrm{pr}\{\varphi(S_j)>1/\gamma\}$}
	\label{sec:conditional_validity}
	We note that since $\varphi(S_j)$ is an $e-$value under $H_{0j}^{r/n}$, $\mathrm{pr}\{\varphi(S_j)>1/\gamma\}\le \gamma$ for all $\gamma\in(0,1)$. However, to establish the validity of $e-$Filter for simultaneous inference under unknown study dependence, a sharper upper bound is required. Proposition \ref{prop:general_n_r} provides this refinement. We then examine the simultaneous error rate control property of $e-$Filter for finite $m$ (Section \ref{sec:theory_m_finite}) and for $m\to\infty$ (Section \ref{sec:theory_m_inf}). Throughout, we employ the calibrator $\varphi(x)=\kappa x^{\kappa-1},~\kappa\in(0,1)$ from Equation (1) of \cite{vovk2021values}.
	
	To facilitate the analysis, we assume that under $H_{0 j}^{r / n}$, $S_j$ and $F_j$ have distributions from the family of Lehmann alternatives \citep{10.1214/aoms/1177729080} with their CDFs satisfying the following relationship:
	\begin{equation}
		\label{eq:lehmann}
		\mathrm{pr}\left(S_j\le x\right)\le x^{d_1} < x^{d_2} \leq \mathrm{pr}(F_j \le x) \quad \text{for each $x \in (0,1)$},
	\end{equation}
	where $d_1\geq 1$, since  $S_j$ is a valid $p$-value under $H_{0 j}^{r / n}$, and $0<d_2<1$. We have the following result.
	\begin{proposition}
		\label{prop:general_n_r}
		Let $\varphi$ denote a calibrator \( \varphi(x) := \kappa x^{\kappa - 1},~\kappa\in(0,1)\) and suppose Equation \eqref{eq:lehmann} holds. Denote $\kappa^{*}=\max(0,d_2-d_1+1)$. Then, for all $\gamma\in(0,1)$ and for all $\kappa\in[\kappa^{*},1)$,
		$$\mathrm{pr}\{\varphi(S_j)\ge 1/\gamma\}\le \gamma~\mathrm{pr}\{\varphi(F_j)\ge 1/\gamma\} \text{ under }H_{0j}^{r/n}.$$
	\end{proposition}
	While Proposition \ref{prop:general_n_r} provides an improved upper bound on $\mathrm{pr}\{\varphi(S_j)\ge 1/\gamma\}$, it is unclear how the strength of dependence across studies impacts $\kappa^*$. We remark on this issue below.
	\begin{figure}[!h]
		\centering
		\includegraphics[width=0.85\linewidth]{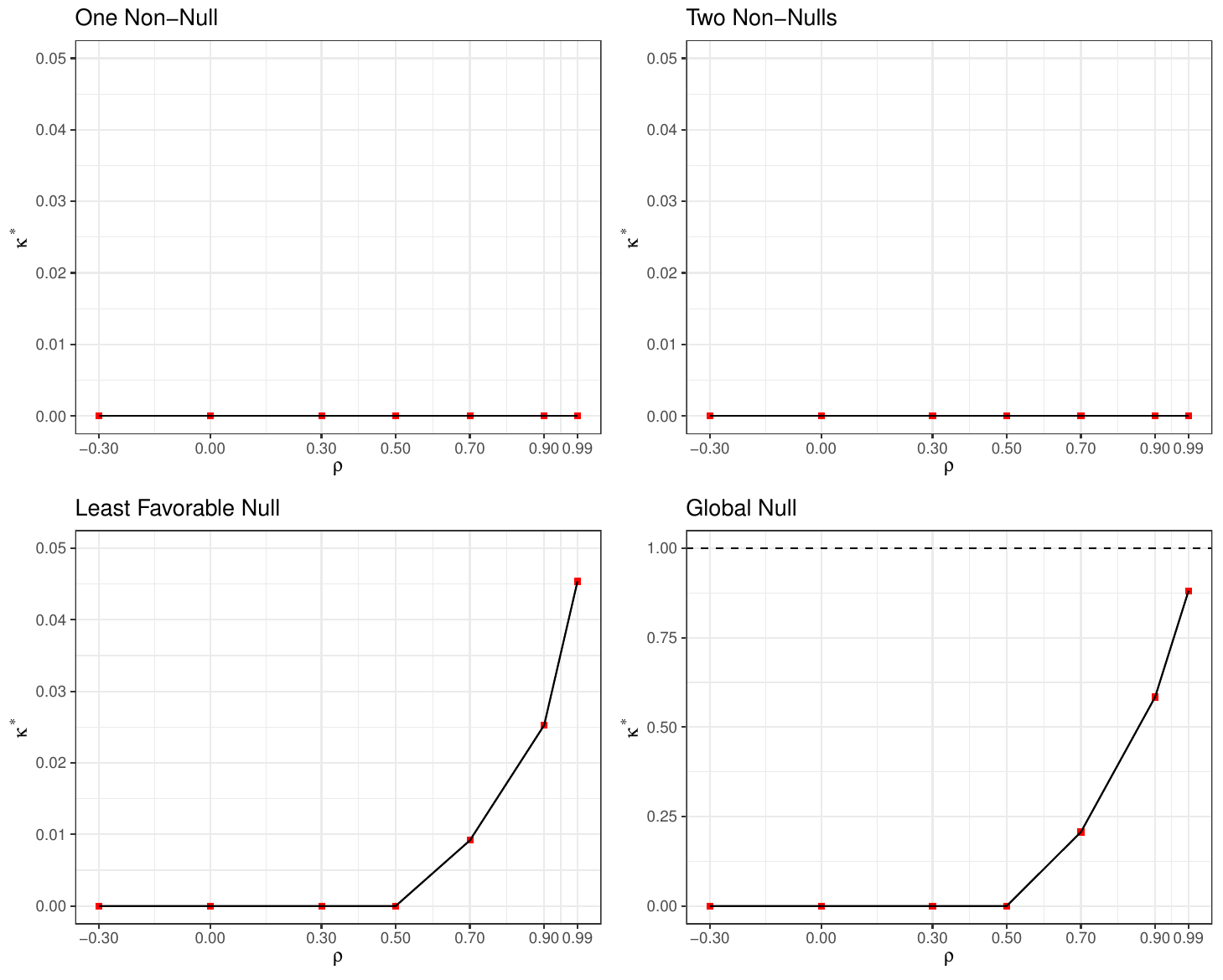}
		\caption{Testing $H_{0j}^{4/4}$ with $P_{ij}=2\Phi(-|X_{ij}|)$. Here $(X_{1j},\ldots,X_{4j})$ are multivariate Normal with mean $\bm\mu$ and covariance matrix $\bm{\Sigma}=(1-\rho)\bm I_4+\rho\bm 1\bm1^T$, and $S_j=P^B_{4/4,j}$, $F_j=P_{(3) j}$. Consider Equation \eqref{eq:lehmann} and denote $\kappa^*=\max(0,d_2-d_1+1)$. Each panel plots $\kappa^{*}$ against different values of $\rho$. Top left: $\bm \mu=(0,0,0,3)^T$. Top right: $\bm \mu=(0,3,0,3)^T$. Bottom left: $\bm \mu=(0,3,3,3)^T$. Bottom right: $\bm \mu=(0,0,0,0)^T$. Analytical calculations are given in Section \ref{app: fig_details} of the supplement.}
		\label{fig:fig_sec_4}
	\end{figure}
	
	Consider testing $H_{0j}^{4/4}$ with $P_{ij}=2\Phi(-|X_{ij}|)$. Here $(X_{1j},\ldots,X_{4j})$ are multivariate Normal with mean $\bm\mu$ and covariance matrix $\bm{\Sigma}=(1-\rho)\bm I_4+\rho\bm 1\bm1^T$. Set $S_j=P^B_{4/4,j}$ and $F_j=P_{(3) j}$. Figure \ref{fig:fig_sec_4} examines four configurations under which $H_{0j}^{4/4}$ is true. 
	The top-left panel plots $\kappa^{*}$ as a function of $\rho$ for the configuration where exactly one base null is false, $\bm{\mu}=(0,0,0,3)^T$. We observe that $\kappa^{*}=0$ across all $\rho$, indicating that in this setting the strength of dependence among studies has no effect on $\kappa^{*}$. The top-right panel corresponds to $\bm{\mu}=(0,3,0,3)^T$ and leads to the same conclusion. The bottom-left panel considers the least favorable null configuration, where three of the four base hypotheses are non-null, $\bm{\mu}=(0,3,3,3)^T$. In this case, $\kappa^{*}$ varies with $\rho$ and Proposition \ref{prop:general_n_r} continues to hold uniformly for all $\rho\ge -1/3$ provided that $\kappa^{*}>0.045$. In contrast, the bottom-right panel, which represents the global null with $\bm{\mu}=(0,0,0,0)^T$, shows that although $\kappa^{*}$ again depends on $\rho$, it must be at least $0.9$ for Proposition \ref{prop:general_n_r} to hold uniformly in $\rho$. Typically, such extreme dependence between the studies is unrealistic in practice because it offers no new scientific information for the underlying replicability analysis. Nevertheless, the preceding discussion suggests that if the studies exhibit almost-perfect dependence and the global null configuration dominates, then the choice of $\kappa$ may have a substantial impact on the validity of $e-$Filter. Analytical calculations underlying Figure \ref{fig:fig_sec_4} are given in Section \ref{app: fig_details} of the supplement.
	\subsection{FWER / PFER and FDR control for finite $m$}
	\label{sec:theory_m_finite}
	We collect the regularity conditions required for our analysis below.
	\begin{assumption}
		\label{assume:1}
		Under $H_{0 j}^{r / n}$, the CDFs of $S_j$ and $F_j$ satisfy the following relationship:
		$$\mathrm{pr}\left(S_j\le x\right)\le x^{d_1} < x^{d_2} \leq \mathrm{pr}(F_j \le x) \quad \text{for each $x \in (0,1)$},$$where $d_1\ge 1$ and $0<d_2<1$.
	\end{assumption}
	\begin{assumption}
		\label{assume:2}
		The $n\times m$ matrix $\mathcal P$ contain column-wise independent valid $p$-values. 
	\end{assumption}
	We re-state Equation \eqref{eq:lehmann} as Assumption \ref{assume:1}, which allows the use of the upper bound developed in Proposition \ref{prop:general_n_r} to establish simultaneous error rate control of $e-$Filter under unknown dependence. Assumption \ref{assume:2} requires that within a study the $p-$values are independent, an assumption also made in \cite{Owen, Dickhaus2024,liang2025powerful}. However, unlike these works, we do not require row-wise independence across $\mathcal P$. 
	While our analysis relies on Assumption \ref{assume:2},  the numerical experiments in Section \ref{sec:sims} demonstrate that $e-$Filter continues to control the error rates at the specified level even when the $p-$values within each study are dependent. 
	
	Denote $\kappa^{*}=\max(0,d_2-d_1+1)$. We have the following guarantee on the control of FWER and PFER when $m$ is finite.
	\begin{theorem}
		\label{thm:e_adafilter_fwer_control}
		Under Assumptions \ref{assume:1}--\ref{assume:2} and for all $\kappa\in[\kappa^*,1)$, the $e-$Filter procedure in Definition \ref{def:eadaBon} controls FWER and PFER at level $\alpha$ for the null hypotheses $\mathcal H^{r/n}$.
	\end{theorem}
	The next theorem shows that $e-$Filter controls FDR at level $\alpha^{1+d_1} C(m)$ where $C(m) = \sum_{j=1}^{m}1/j$. 
	\begin{theorem}
		\label{thm:e_adafilter_fdr_control}
		Denote $C(m) = \sum_{j=1}^{m}1/j$. Then, under Assumptions \ref{assume:1}--\ref{assume:2} and for all $\kappa\in[\kappa^*,1)$, the $e-$Filter procedure in Definition \ref{def:eadaBH} controls FDR at level $\alpha^{1+d_1} C(m)$ for the null hypotheses $\mathcal H^{r/n}$.
	\end{theorem}
		The proof of Theorem \ref{thm:e_adafilter_fdr_control} closely follows the proof technique of Theorem 4.3 of \cite{Owen}, where an inflation factor $C(m)$ similarly arises in establishing FDR control for finite $m$. However, the numerical experiments in Section \ref{sec:sims} suggest that $e-$Filter can control FDR at level $\alpha$ and remains more powerful than competing approaches designed for testing $\mathcal H^{r/n}$ with dependent studies. In Section \ref{sec:theory_m_inf}, we show that $e-$Filter controls the FDR at level $\alpha$ when $m\to\infty$. 
		\subsection{Asymptotic FDR control}
		\label{sec:theory_m_inf}
		To prove asymptotic FDR control of $e-$Filter (Definition \ref{def:eadaBH}), we borrow the regularity conditions from \cite{Owen} and adapt them to our setting. Specifically, the following technical conditions are required.
		\begin{assumption}
			\label{asymp_assump_2}
			Denote $\mathcal H_{0}^{r/n}=\{j:H_{0j}^{r/n}~\text{is true.}\}$ and let $m_0$ be its cardinality. Similarly, let $\mathcal H_{1}^{r/n}=\{j:H_{0j}^{r/n}~\text{is false.}\}$ and suppose $m_1$ is its cardinality. We assume that the following limits exist:
			$$
			\begin{aligned}
				\lim _{m \rightarrow \infty} \frac{m_0}{m} & =\pi_0 \in(0,1), \\
				\lim _{m \rightarrow \infty} \frac{1}{m_0} \sum_{j \in \mathcal{H}_0^{r / n}} \mathrm{pr}\left(F_j<\gamma\right)  =\tilde{F}_0(\gamma), &~
				\lim _{m \rightarrow \infty} \frac{1}{m_1} \sum_{j \in \mathcal{H}_1^{r / n}} \mathrm{pr}\left(F_j<\gamma\right)  =\tilde{F}_1(\gamma), \\
				\lim _{m \rightarrow \infty} \frac{1}{m_0} \sum_{j \in \mathcal{H}_0^{r / n}} \mathrm{pr}\left(S_j<\gamma\right)  =\tilde{S}_0(\gamma), &~
				\lim _{m \rightarrow \infty} \frac{1}{m_1} \sum_{j \in \mathcal{H}_1^{r / n}} \mathrm{pr}\left(S_j<\gamma\right) =\tilde{S}_1(\gamma) .
			\end{aligned}$$
		\end{assumption}
		For a fixed $n$, there are $2^n$ possible configurations of true and false base null hypotheses. Denote by $\mathcal{A}\in\{0,1\}^n$ any such configuration, and let $m_{\mathcal{A}}$ denote the number of PC hypotheses corresponding to configuration $\mathcal{A}$. Let $\mathcal{H}_{0i}$ and $\mathcal{H}_{1i}$ represent the sets of true nulls and true non-nulls for the $i$th study, respectively.
		Then, Assumption~\ref{asymp_assump_2} holds if: (i)  $\lim_{m \to \infty} m_{\mathcal{A}}/{m}$ exists for every $\mathcal{A}$, and (ii) for each $i$, the $p$-values $\{P_{ij} : j \in \mathcal{H}_{0i}\}$ are identically distributed across $j$ and similarly, $\{P_{ij} : j \in \mathcal{H}_{1i}\}$ are identically distributed across $j$.
		
		Under Assumption \ref{asymp_assump_2}, denote
		$\tilde{F}(\gamma)=\pi_0 \tilde{F}_0(\gamma)+\left(1-\pi_0\right) \tilde{F}_1(\gamma)$ and $\tilde{S}(\gamma)=\pi_0 \tilde{S}_0(\gamma)+\left(1-\pi_0\right) \tilde{S}_1(\gamma)$. For a given $\gamma$, we adapt the definition of ``asymptotic FDR'' from \cite{Owen} as follows:
		$$
		f_{e}^{\infty}(\gamma)= \begin{cases}\dfrac{\gamma \tilde{F}\{\varphi^{-1}(1/\gamma)\}}{\tilde{S}\{\varphi^{-1}(1/\gamma)\}}, & \text { if } \tilde{S}\{\varphi^{-1}(1/\gamma)\}>0 \\ 0, & \text {otherwise}\end{cases}.
		$$ 
		Let $
		\gamma_e^{\infty}=\sup \left\{\gamma: f_e^{\infty}(\gamma) \leq \alpha\right\}$. Then $f_e^{\infty}(\gamma)$ is 0 for $\gamma=0$ and exceeds 1 when $\gamma=1$. So, the above set is nonempty. We need the following technical assumption on the functions $f_e^{\infty}(\cdot), \tilde{S}_0(\cdot)$ and $\tilde{S}_1(\cdot)$ around $\gamma_e^{\infty}$:
		\begin{assumption}
			\label{asymp_assump_3}
			The following conditions hold: (a) There exists $\delta>0$ such that $f_e^{\infty}(\gamma)$ is monotonically increasing in the interval $\left(\gamma_e^{\infty}-\delta, \gamma_e^{\infty}\right]$, and (b) $\tilde{S}_0(\gamma)$ and $\tilde{S}_1(\gamma)$ are both continuous at $\gamma_e^{\infty}$.  
		\end{assumption}
		The first condition is required in order to guarantee that the limit of the threshold $\gamma_{e}^{\sf FDR}$ 
		is unique when $m \to \infty$. The second condition is satisfied if there are sufficient points  around $\gamma_{e}^{\infty}$ when $m$ is large.
		
		We have the following theorem on the asymptotic FDR control of $e-$Filter.
		\begin{theorem}
			\label{thm:asymp_e_filter_fdr_control}
			Under Assumptions \ref{assume:1}--\ref{asymp_assump_3} and for all $\kappa\in[\kappa^*,1)$, the following holds for the $e-$Filter procedure in Definition \ref{def:eadaBH} as $m\to\infty$:
			$$
			\begin{aligned}
				\gamma_e^{\sf FDR}\stackrel{p}{\to}\gamma_e^{\infty}~\text{and}~
				\mathrm{FDP}\stackrel{p}{\to}\dfrac{\pi_0\tilde{S}_0\{\varphi^{-1}(1/\gamma_e^{\infty})\}}{\tilde{S}\{\varphi^{-1}(1/\gamma_e^{\infty})\}}\le \alpha.
			\end{aligned}$$
			Thus, the $e-$Filter FDR procedure controls FDR at level $\alpha$ for the null hypotheses $\mathcal H^{r/n}$ as $m\to\infty$.
		\end{theorem}
		\section{Numerical experiments: FDR control}
		\label{sec:sims}
		We evaluate the following six procedures on simulated data for FDR control at the nominal level $\alpha = 0.2$: (1) {AdaFilter}, which tests PC nulls under the assumption of study independence (we use the implementation available at \url{https://github.com/jingshuw/adaFilter}), (2) {BH -$P_{r/n,j}^B$}, a baseline procedure that applies the BH procedure on the Bonferroni PC $p-$values ($P^B_{r/n}$), (3) {BH -$P_{r/n}^C$}, another baseline procedure that applies the BH procedure on the Cauchy PC $p-$values ($P^C_{r/n,j}$), (4) the $e-$PCH procedure from Section \ref{sec:e-PCH}, (5) $e-$Filter B with $S_j=P_{r/n,j}^B,~F_j=(n-r+1)P_{(r-1)j}$, and (6) $e-$Filter C with $S_j=P_{r/n,j}^C,~F_j= \tan\{(0.5-P_{(r-1)j})\pi\}$. For the three $e$-value–based methods, we use the calibrator $\varphi(x) = \kappa x^{\kappa-1}$ with $\kappa \in (0,1)$. The tuning parameter $\kappa$ is selected from the discrete grid $\{0.01, \ldots, 0.09, 0.1, \ldots, 0.9\}$ as the smallest value that yields the largest number of rejections for the corresponding procedure.
		The six methods are evaluated across five different simulation settings, to be described subsequently, and for each setting we adopt the data generation scheme of \cite{Owen}. Specifically, we set $m=10,000$ and consider the following six configurations of $n$ and $r$: $(n,r)\in\{(2,2),~(4,2),~(8,2),~(4,4),~(8,4),~(8,8)\}$. Note that for each $n$ there are $2^n$ combinations of the base hypotheses $\mathcal H_{0j}$. So, denote $\pi_{00}$ as the probability of the global null combination and $\pi_1$ as the probability of the combinations where $H_{0j}^{r/n}$ is false. All remaining combinations where $H_{0j}^{r/n}$ is true have equal probabilities that sum to $1-\pi_{00}-\pi_1$. In our simulations, we fix $\pi_1=0.01$.
		\\[0.3ex]
		\noindent\textbf{Scenario 1 (equicorrelated studies) - }we sample the $m\times n$ matrix of test statistics $\bm X$ from a matrix Normal distribution with $m\times n$ mean matrix $\bm\mu=(\mu_{ij}:1\le i\le n,~1\le j\le m)$, $m\times m$ row covariance matrix $\bm \Sigma_1=(\Sigma_1)_{i,j}$ and $n\times n$ column covariance matrix $\bm \Sigma_2$. In this scenario, $\Sigma_{1,ij}=0.5^{|i-j|}$, $\bm \Sigma_2=\rho\bm 1\bm 1^T+(1-\rho)\bm I_n$ and $\pi_{00}=0.98$. We test $H_{0ij}:\mu_{ij}=0~vs~H_{1ij}:\mu_{ij}\ne 0$ and calculate $P_{ij}=2\Phi(-|X_{ij}|)$. When $H_{0ij}$ is false, we set $\mu_{ij}\in\{-6,-5,-4,4,5,6\}$ with equal probability. 
		\begin{table}[!h]
			\centering
			\scalebox{0.8}{\begin{tabular}{lcccccccc}
					\hline
					& \multicolumn{2}{c}{$\rho = 0.2$} & \multicolumn{2}{c}{$\rho = 0.4$} & \multicolumn{2}{c}{$\rho = 0.6$} & \multicolumn{2}{c}{$\rho = 0.8$} \\
					\hline
					Method      & FDR & Recall & FDR & Recall & FDR & Recall & FDR & Recall \\
					\hline
					AdaFilter     &  0.115    &   0.983    &   0.257   &   0.983    &  0.478    &   0.985    &  0.708    &   0.990    \\
					BH-$P^{B}_{r/n}$      &  0.002    &   0.756    &    0.004  &    0.761   &  0.011    &  0.770     &   0.029   &  0.782    \\
					BH-$P^{C}_{r/n}$       & 0.002    &   0.758    &   0.005   &    0.764   &   0.016   &   0.773    &  0.047    &   0.785    \\
					\hline
					$e-$PCH   &   $<0.001$   &   0.432    &  $<0.001$    &    0.432   &   $<0.001$    &   0.433    &   0.001    &  0.434     \\
					$e-$Filter C   &    0.005  &   0.892    & 0.013     &    0.893   &   0.044   &   0.895   &  0.136    &   0.900 \\
					$e-$Filter B   &    0.006  &   0.896    & 0.015     &    0.898   &   0.042   &   0.900    &  0.112    &   0.905 
					\\
					\hline
			\end{tabular}}
			\caption{Scenario 1: Average FDR and Recall for different methods targeting a nominal FDR of $\alpha=0.2$. Results for each $n$ and $r$ are presented in Figure \ref{fig:scenario_1_fdr_rbyn} of the supplement.}
			\label{tab:setting_1_fdr}
		\end{table}
		
		For $\rho\in\{0.2,0.4,0.6,0.8\}$, Table \ref{tab:setting_1_fdr} reports the average FDR and Recall across $B=500$ repetitions of the data generation scheme and six $(n,r)$ combinations. Results for each $n$ and $r$ are presented in Figure \ref{fig:scenario_1_fdr_rbyn} of the supplement. With the exception of AdaFilter, all methods control FDR at the nominal level for all values of $\rho$. The two variants of $e-$Filter have similar power with $e-$Filter B marginally more powerful than $e-$Filter C. Furthermore, $e-$PCH is the most conservative procedure in this setting. While $e-$Filter dominates the PC $p-$value approaches and $e-$PCH in power, Figure \ref{fig:scenario_1_fdr_rbyn} reveals that its gain in power is more substantial when $n=8$. 
		\\[0.3ex]
		\noindent\textbf{Scenario 2 (Common subjects) - }this scenario mimics the real data application and reports the impact of common subjects on the performance of various methods. We consider a setting where $n=5,~r=2$ and test $H_{0ij}:\mu_{ij}=0~vs~H_{1ij}:\mu_{ij}\ne 0$. We fix $\pi_{00}=0.98$ and set $\mu_{ij}\in\{-4,-3,3,4\}$ with equal probability whenever $H_{0ij}$ is false. For study 1,  $X_{1jk}=\mu_{1j}+\epsilon_{1jk}$ where $\epsilon_{1jk}\stackrel{i.i.d.}{\sim} N(0,s),~k=1,\ldots,s$. Thus, study 1 involves $s$ subjects and $X_{1j}=(1/s)\sum_{k=1}^{s} X_{1jk}\sim N(\mu_{1j},1)$. We fix $s=5000$. Studies 2, 3, and 4 do not share subjects amongst themselves, but rely on a mutually exclusive set of $1000$ subjects each from study 1. Specifically, for $i\in\{2,3,4\}$ and $k=1,\ldots,1000$, $X_{ijk}=\mu_{ij}+\epsilon_{1j,k+1000(i-2)}/\sqrt{5}$. Finally, study 5 involves $1000$ subjects and borrows $q$ of them from study 2. So,  $X_{5jk}=\mu_{5j}+\epsilon_{5jk},~\epsilon_{5jk}\stackrel{i.i.d.}{\sim}N(0,1000)$ for $k=1,\ldots,1000-q$ and $X_{5jk}=\mu_{5j}+\epsilon_{1j,k-1000+q}$ for $k=1000-q+1,\ldots,1000$.
		\begin{figure}[!h]
			\centering
			\includegraphics[width=0.8\linewidth]{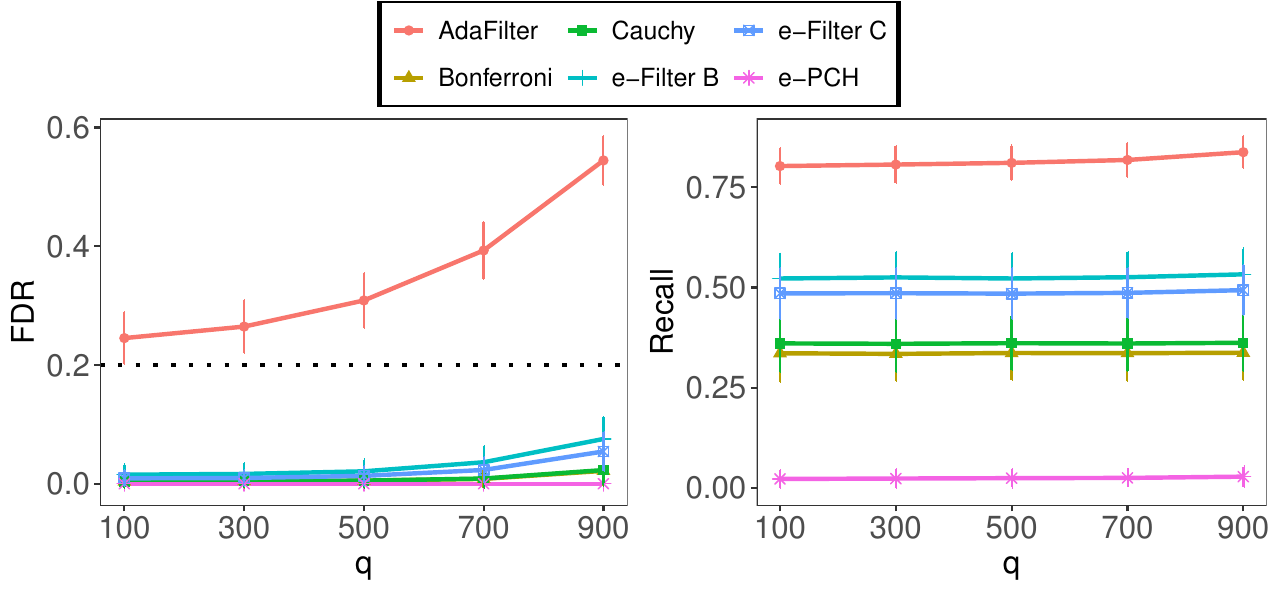}
			\caption{Scenario 2: Average FDR and Recall for different methods as $q$ varies. Here $\alpha=0.2$. {The error bars represent one standard deviation above and below the average FDR (Recall) from $B$ repetitions.}}
			\label{fig:scenario_3_fdr}
		\end{figure}
		
		Figure \ref{fig:scenario_3_fdr} reports the average FDR and Recall for different methods as $q$ varies. Common subjects induce positive dependence between the study-specific summary statistics and in this setting AdaFilter fails to control FDR at the nominal level for all values of $q$. While all other methods control FDR, the two $e-$Filter variants exhibit the highest power. 
		\\[0.5ex]
		\noindent\textbf{Scenario 3 (Negative study dependence) - }here we explore the impact of negative study dependence on the performance of various methods. Specifically, we borrow the data generating scheme from Scenario 1 but test $H_{0ij}:\mu_{ij}=0~vs~H_{1ij}:\mu_{ij}> 0$ and calculate $P_{ij}=1-\Phi(X_{ij})$. Furthermore, we set $\Sigma_{2,ij}=\rho^{|i-j|}$ where $\rho\in\{-0.2,-0.4,-0.6,-0.8\}$. Table \ref{tab:setting_2_fdr} reports the average FDR and Recall across $B$ repetitions of the data generation scheme and six $(n,r)$ combinations. Figure \ref{fig:scenario_2_rbyn} in the supplement reports the results for each $n$ and $r$. Unlike Scenario 1, AdaFilter controls FDR at the nominal level for all but the smallest value of $\rho$. Moreover, for $\rho\in\{-0.6,-0.8\}$, Figure \ref{fig:scenario_2_rbyn} reveals that it fails for the $2/4$ and $4/4$ configurations. In contrast, all other methods control FDR at the nominal level. In terms of power, AdaFilter dominates across all $r/n$ configurations and is closely followed by $e-$Filter B.
		\begin{table}[h!]
			\centering
			\scalebox{0.8}{\begin{tabular}{lcccccccc}
					\hline
					& \multicolumn{2}{c}{$\rho = -0.2$} & \multicolumn{2}{c}{$\rho = -0.4$} & \multicolumn{2}{c}{$\rho = -0.6$} & \multicolumn{2}{c}{$\rho = -0.8$} \\
					\hline
					Method      & FDR & Recall & FDR & Recall & FDR & Recall & FDR & Recall \\
					\hline
					AdaFilter     &  0.041    &   0.250    &   0.047   &   0.250    &  0.086    &   0.251    &  0.220    &   0.252    \\
					BH-$P^{B}_{r/n}$      &  $<0.001$    &   0.194    &    $<0.001$   &    0.194   &  0.001    &  0.194     &   0.004   &  0.194    \\
					BH-$P^{C}_{r/n}$       & $<0.001$    &   0.114    &   $<0.001$   &    0.114   &  $<0.001$   &  0.113    &  0.001   &   0.112    \\
					\hline
					$e-$PCH   &   $<0.001$   &   0.119    &  $<0.001$    &    0.118   &   $<0.001$    &   0.117    &   $<0.001$    &  0.116    \\
					$e-$Filter C   &    0.002  &   0.129    & 0.001     &    0.128   &   0.002   &   0.127    &  0.002    &   0.127    \\
					$e-$Filter B   &    0.002  &   0.224    & 0.002     &    0.224   &   0.005   &   0.224    &  0.020    &   0.224    \\
					\hline
			\end{tabular}}
			\caption{Scenario 3: Average FDR and Recall for different methods targeting a nominal FDR of $\alpha=0.2$. Results for each $n$ and $r$ are presented in Figure \ref{fig:scenario_2_rbyn} of the supplement.}
			\label{tab:setting_2_fdr}
		\end{table}
		\\[0.5ex]
		\noindent\textbf{Scenario 4 (Common controls) - }we consider the following setting: each study is a case–control design that compares the genomes of individuals with a specific disease (‘cases’) to those of healthy individuals (‘controls’) to identify genetic variants, such as SNPs, that occur more frequently between cases and thus can be associated with the disease. To reduce costs, studies often employ a common set of subjects that serve as `controls' \citep{LinSullivan} while the `cases' remain distinct in all studies. We borrow the setting of Scenario 1 but sample the test statistics $X_{ij}$ as follows. We set $X_{ij}=(X_{ij}^{\sf tr}-X_j^{\sf c})/\sqrt 2$ where $X_{ij}^{\sf tr}\stackrel{ind.}{\sim}N(\mu_{ij},1)$ represents the average effect size for the `cases' in study $i$ and base null hypothesis $j$ while $X_{j}^{\sf c}\stackrel{i.i.d.}{\sim}N(0,1)$ denotes the average effect size of the `controls' that are shared by the $n$ studies. 
		
		For $\pi_{00}\in\{0.98,0.8\}$, Table \ref{tab:scenario_4_fdr} reports the average FDR and Recall across $B$ repetitions of the data generation scheme and six $(n,r)$ combinations. Figure \ref{fig:scenario_4_fdr_rbyn} reports the results for each $n$ and $r$. When $\pi_{00}=0.98$, the $e-$Filter variants are the most powerful methods that, unlike AdaFilter, also control the FDR at the nominal level. However, at $\pi_{00}=0.8$, the PC $p-$value methods are more powerful than $e-$Filter. In settings where the proportion of the global null configuration is relatively small, this trend reversal in power is expected since the filtering scheme underlying methods such as $e-$Filter and AdaFilter, is most effective in enhancing power when $\pi_{00}$ overwhelmingly dominates the proportion of all other base null configurations \citep{liang2025powerful}. 
		\begin{table}[h!]
			\centering
			\scalebox{0.8}{\begin{tabular}{lcccc}
					\hline
					& \multicolumn{2}{c}{$\pi_{00} = 0.98$} & \multicolumn{2}{c}{$\pi_{00} = 0.8$} \\
					\hline
					Method      & FDR & Recall & FDR & Recall \\
					\hline
					AdaFilter     &  0.377    &   0.793    &   0.214   &   0.642 \\
					BH-$P^{B}_{r/n}$      &  0.005    &   0.196    &    0.014   &    0.195\\
					BH-$P^{C}_{r/n}$       & 0.007   & 0.206  &   0.016   &  0.206  \\
					\hline
					$e-$PCH   &   $<0.001$   &   0.024    &  $<0.001$    &    0.025  \\
					$e-$Filter C   &    0.024  &   0.355    & 0.009     &    0.177     \\
					$e-$Filter B   &    0.027  &   0.376    & 0.010     &    0.188     \\
					\hline
			\end{tabular}}
			\caption{Scenario 4: Average FDR and Recall for different methods targeting a nominal FDR of $\alpha=0.2$. Results for each $n$ and $r$ are presented in Figure \ref{fig:scenario_4_fdr_rbyn} of the supplement.}
			\label{tab:scenario_4_fdr}
		\end{table}
		\\[0.5ex]
		\noindent\textbf{Scenario 5 (Scale mixture of Normals) - }this is a setting where the $p-$values for the base null hypotheses are calculated from a misspecified null distribution. Specifically, we borrow the setting of Scenario 1 but when the base null $H_{0ij}$ is true, $X_{ij}$ is now a scale mixture of Normals instead of a standard Normal. However, studies continue to rely on the $p-$values $P_{ij}=2\Phi(-|X_{ij}|)$ to test the base nulls. We sample $Y_{ij}\stackrel{i.i.d}{\sim} 0.5~N(0,1)+0.25~N(0,2)+0.25~N(0,4)$ and denote $\bm Y$ as the $n\times m$ matrix with entries $Y_{ij}$. From scenario 1, suppose $\bm \Sigma_1=\bm U\bm U^T$ and $\bm \Sigma_2=\bm V\bm V^T$ denote the Cholesky decompositions of $\bm \Sigma_1$ and $\bm \Sigma_2$. Then we set $\bm X=\bm \mu+\bm V\bm Y\bm U^T$. In particular, this implies $\text{Corr}(X_{ij},X_{i'j})=\rho,~i\ne i'$ and $\text{Corr}(X_{ij},X_{ij'})=0.5^{|j-j'|},~j\ne j'$.
		
		Table \ref{tab:setting_5_fdr} reports the average FDR and Recall from $B$ repetitions of the data generation scheme and six $(n,r)$ combinations. Figure \ref{fig:scenario_5_rbyn} report the results for each $n$ and $r$. Table \ref{tab:setting_5_fdr} suggests that all methods, with the exception of AdaFilter, control FDR at the nominal level; however, Figure \ref{fig:scenario_5_rbyn} shows that the PC $p-$value methods and the $e-$Filter variants fail for the $4/4$ configuration. Furthermore, in terms of power, $e-$Filter dominates the PC $p-$value methods primarily when $n=8$. 
		\begin{table}[!h]
			\centering
			\scalebox{0.8}{\begin{tabular}{lccccccccccc}
					\hline
					&\multicolumn{2}{c}{$\rho = 0$} &
					\multicolumn{2}{c}{$\rho = 0.2$} & \multicolumn{2}{c}{$\rho = 0.4$} & \multicolumn{2}{c}{$\rho = 0.6$} & \multicolumn{2}{c}{$\rho = 0.8$} \\
					\hline
					Method   & FDR & Recall   & FDR & Recall & FDR & Recall & FDR & Recall & FDR & Recall \\
					\hline
					AdaFilter   &0.345 & 0.916 &  0.352    &   0.921    &   0.398   & 0.932      &   0.498   &   0.950    &  0.661   & 0.972   \\
					BH-$P^{B}_{r/n}$   & 0.160& 0.666  &  0.144    &   0.671  &  0.114    & 0.683      &  0.084    &  0.704     &   0.062   & 0.740   \\
					BH-$P^{C}_{r/n}$    &0.162 & 0.670  & 0.146    &   0.674    &  0.120    &  0.687     &   0.095   &  0.708     &    0.085  &  0.744  \\
					\hline
					$e-$PCH  &0.031 &0.403 &  0.025   &   0.404    &   0.016   &  0.406  & 0.008   & 0.411      &   0.003   & 0.419   \\
					$e-$Filter C &0.098 &0.755  &    0.091  &   0.763    &  0.080    &  0.782     &  0.081    &  0.811     &   0.127   &  0.852  \\
					$e-$Filter B & 0.115& 0.762 &    0.108  &   0.770   &  0.096    &   0.790    &    0.094  &   0.819    & 0.123     &  0.859
					\\
					\hline
			\end{tabular}}
			\caption{Scenario 5: Average FDR and Recall for different methods targeting a nominal FDR of $\alpha=0.2$. Results for each $n$ and $r$ are presented in Figure \ref{fig:scenario_5_rbyn} of the supplement.}
			\label{tab:setting_5_fdr}
		\end{table}
		
		In Section \ref{app:pfer_sims} of the supplement, we report the performance of these methods for controlling PFER at $\alpha=1$.
		\section{Replicability of GWAS for LDL-C}
		\label{sec:realdata}
		Annually, cardiovascular diseases (CVDs) account for roughly one-third of all deaths worldwide and remain a major public health challenge, especially in low- and middle-income regions. Cholesterol levels are key risk factors for CVDs and include high-density lipoprotein cholesterol (HDL-C), low-density lipoprotein cholesterol (LDL-C), triglycerides, and total cholesterol. Elevated LDL-C levels, in particular, are a major and causal determinant of CVD risk \citep{sandhu2008ldl}. Although the role of family history in CVD risk has long been recognized, GWAS have illuminated the genetic variants underlying these diseases \citep{sandhu2008ldl}. 
		
		In this section, we investigate the replicability of GWAS for LDL-C, and consider the following $n=5$ genetic studies: European study (EUR2010) \citep{teslovich2010biological}, UK Biobank (UKB) study \citep{bycroft2018uk}, Biobank Japan (BBJ) study \citep{nagai2017overview}, within-family GWAS consortium (FGC) study \citep{howe2022within} and the global lipids genetics consortium (GLGC) study \citep{graham2021power}. The EUR2010 study focused on European ancestries. In this study, the discovery GWAS data had more than 100,000 individuals of European descent from 46 studies across the US, Europe, and Australia. The UKB study mainly centers around Europeans, dominated by White-British individuals. The BBJ study is based on the Japanese (East Asian) population. The FGC study is primarily European and combines the data from 18 European cohorts. The bigger GLGC study comprises 1.65 million individuals from different ancestries. The ancestry configuration is dominated by Europeans $(80\%)$ with limited contributions from other ancestries, such as East Asian, African, Hispanic, and South Asian.
		
		\begin{figure}[!h]
			\centering
			\includegraphics[width=0.6\linewidth]{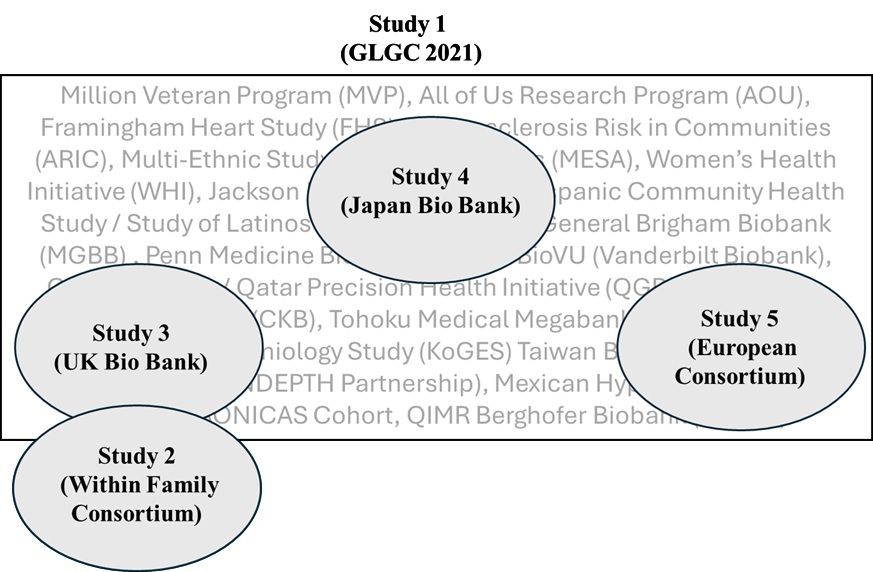}
			\caption{The five studies have a varying extent of sample overlap. While the UKB and BBJ studies are non-overlapping, the FGC and UKB studies have a partial sample overlap $(\text{around } 40,000)$.
				The GLGC study has the largest sample size among these studies and includes EUR2010, UKB, BBJ and the Million Veteran Program among others. However, the GLGC study did not include the FGC cohort. Image source: \protect\url{https://csg.sph.umich.edu/willer/public/glgc-lipids2021/}}
			\label{fig:five_study}
		\end{figure}
		As Figure \ref{fig:five_study} suggests, these studies have a varying extent of sample overlap. While the UKB and BBJ studies are non-overlapping, the FGC and UKB studies have a partial sample overlap (around $40,000$. See Table 1 row 16 at \url{https://www.withinfamilyconsortium.com/cohorts/}). The GLGC study has the largest sample size among these studies \citep{teslovich2010biological}. It utilized the EUR2010 study as its baseline and expanded to include other datasets and ancestries such as UKB, BBJ, the Million Veteran Program (MVP), etc. However, the GLGC study did not include the FGC cohort. Therefore, unlike the remaining three studies, the FGC cohort is not a complete subset of the GLGC study. Thus, these five studies exhibit sample overlap and induce dependence between their summary statistics. The data sources for the five studies are provided in Section \ref{app:data_source} of the supplement. 
		
		Since LDL-C is a polygenic trait, numerous genetic variants influence its levels in the population. Among these variant types, we focus on single nucleotide polymorphisms (SNPs). In each study, the $j^{th}$ base null hypothesis states that LDL-C is not associated with the genotype of SNP $j$. Our analysis relies on the $p$-values of these associations for each study. Typically, the $p$-values are obtained from linear regressions of LDL-C on SNP genotypes, adjusting for relevant covariates such as sex, age, and principal components (PCs) of genetic ancestry. Consequently, the publicly available summary statistics and corresponding $p$-values are derived after accounting for these key covariates.
		Since the total number of SNPs is huge, we work with an important class which captures the correlation structure of the SNPs, namely the set of HapMap SNPs, considered by the International HapMap Project \citep{Altshuler2005}. This yields a collection of $m = 459{,}092$ SNPs with base $p$-values available across all five studies. The base-pair positions of these SNPs are aligned according to the GRCh37 (Genome Reference Consortium Human Build 37) assembly for consistency across studies.
		
		In our analysis, we test $\mathcal H^{r/5}$ at the nominal FDR level $\alpha=1\%$ and consider the following three methods from Section \ref{sec:sims}: $e-$Filter B, AdaFilter and Bonferroni (BH -$P_{r/n,j}^B$). Table \ref{tab:realdata-rbyn} reports the number of SNPs replicated by these methods for different values of $r$. 
		\begin{table}[!h]
			\centering
			\scalebox{0.8}{\begin{tabular}{cccc}
					$r$ & AdaFilter & $e-$Filter B & Bonferroni  \\
					\hline
					2 &    4,381       &  2,210        &  1,955 \\
					3 &   1,760        &    980      &   723  \\
					4 &    917       &     572     &   343   \\
					5 &   290        &      152    &   78  \\
					\hline
			\end{tabular}}
			\caption{The number of SNPs replicated by different methods for different values of $r$  at the nominal FDR level $\alpha=1\%$.}
			\label{tab:realdata-rbyn}
		\end{table}
		Although AdaFilter provides close to twice the number of discoveries as $e-$Filter B and Bonferroni, dependent studies can inflate the FDR for AdaFilter, as seen in Section \ref{sec:sims}. Since the total number of replicating SNPs is in the order of thousands, it is impractical to verify whether, for a given $r$, each replicating SNP set is already known in various databases, e.g., the NHGRI-EBI GWAS catalog. So, we consider an aggregate-level approach to further validate the three sets of replicating signals from a biological point of view. 
		
		We fix $r=2$ and map the significant SNPs identified by each method to genes based on physical proximity. Each SNP may either lie within a gene or in its surrounding region. Such annotation of associated SNPs to nearby genes is crucial, as it enables the construction of a gene subset for pathway enrichment 
		and downstream functional analyses. These proximal genes may influence the corresponding phenotype through biologically meaningful mechanisms. We use the Variant Effect Predictor (VEP) tool \citep{mclaren2016ensembl} to map genes to the replicated SNPs, yielding three gene sets corresponding to $e-$Filter B, AdaFilter, and Bonferroni. To identify biologically relevant pathways, we then perform pathway enrichment analysis using the Reactome method \citep{milacic2024reactome} available within the Enrichr tool \citep{kuleshov2016enrichr}. This analysis highlights pathways (i.e., gene sets) that are significantly enriched for these genes, revealing biologically meaningful mechanisms underlying the associations.
		\begin{table}[!t]
			\centering
			\scalebox{0.8}{\begin{tabular}{lcc|cc}
					\hline
					& \multicolumn{2}{c|}{Ratio of Odds Ratios}    & \multicolumn{2}{c}{Ratio of Combined Scores}   \\
					\hline
					Pathway & $e-$Filter & Bonferroni & $e-$Filter & Bonferroni  \\
					\hline
					Plasma Lipoprotein Assemb, Remodel, and Clear.                            & 1.662      & 1.343    & 2.078     & 1.171     \\
					Plasma Lipoprotein Assembly                                                         & 1.388     & 0.826    & 1.443     & 0.500    \\
					Plasma Lipoprotein Remodeling                                                       & 1.467     & 1.303    & 1.565
					& 1.131
					\\
					Plasma Lipoprotein Clearance                                                        & 1.873	& 1.417	&	2.682	& 1.335
					\\
					Chylomicron Assembly                                                                & 0.906	& 0.451	&	0.742	& 0.190
					\\
					Chylomicron Remodeling                                                              & 0.794	& 0.410	&	0.677	& 0.206
					\\
					Chylomicron Clearance                                                               & 1.040	& 0.001	&	1.322	& 0.001
					\\
					NR1H2 and NR1H3-mediated Signaling                                                  & 1.714	& 1.999	&	2.388 &	3.095
					\\
					NR1H3 \& NR1H2 Regulate Gene   Expression & 1.636	& 1.907	&	2.111	& 2.729
					\\
					HDL Remodeling    & 2.119 &	1.094 &		2.942 &	0.895
					\\
					\hline
			\end{tabular}}
			\caption{Reactome pathway enrichment analysis using $n=5$ studies, $r=2$ and FDR $\alpha = 0.01$.}
			\label{tab:pathway_1}
		\end{table}
		
		We discuss ten common LDL-C related pathways. Table \ref{tab:pathway_1} lists these pathways identified by Reactome. 
		For instance, the plasma lipoprotein assembly pathway is related with LDL cholesterol \citep{vance1990assembly}. It corresponds to packaging fats, such as triglycerides and cholesterol, into particles known as very-low-density lipoproteins (VLDL). Similarly, plasma lipoprotein remodeling is relevant to LDL-C \citep{gibbons2004synthesis}. To compare the extent or strength of enrichment of the three sets of genes in each pathway, we consider two statistical measures: the odds ratio of enrichment and the combined score, both of which are reported by the Enrichr platform. Informally, the odds ratio measures the magnitude of enrichment by comparing the odds of pathway membership among input genes to that in the background population. So, an odds ratio in excess of 1 indicates that input genes from the given pathway appear more often than expected. The combined score, on the other hand, combines the statistical significance from the $p-$value with a $z-$score that quantifies how much the observed overlap between the input gene list and a pathway deviates from that expected by random chance, yielding a single measure that reflects both the magnitude and significance of enrichment. Section \ref{app:reactome} of the supplement formally describes these measures in our context.
		
		Table \ref{tab:pathway_1} reports the ratio of odds ratios and the ratio of combined scores, which measure how much stronger (or weaker) each pathway is enriched under $e-$Filter B (Bonferroni) relative to AdaFilter. The ratio of odds ratios reveal that among the ten pathways, the $e-$Filter B derived gene set exhibits stronger enrichment in eight pathways relative to AdaFilter, whereas the Bonferroni derived set shows stronger enrichment in six. When comparing $e-$Filter B and Bonferroni, $e-$Filter B is relatively more enriched in six pathways. The ratio of combined scores lead to similar conclusions. 
		Overall, these findings indicate that while AdaFilter identifies over twice as many replicated SNPs as $e$-Filter B, the SNPs replicated by $e$-Filter B exhibit greater biological coherence, underscoring the functional replicability of its discoveries and their closer alignment with established mechanisms of LDL-C regulation. 
		\section{Discussion}
		\label{sec:discuss}
		We conclude this article with a discussion on the limitations of $e-$Filter and avenues for future research.
		
		We discuss two main limitations of $e-$Filter. \textit{First}, $e-$Filter serves as a powerful alternative to PC $p$-value–based methods when the global null configuration overwhelmingly dominates all others. However, when the proportion of global nulls ($\pi_{00}$) is relatively small, $e-$Filter may exhibit lower power than PC $p$-value–based approaches, since its filtering mechanism is most effective when $\pi_{00}$ is the predominant configuration. This scenario is examined in simulation Setting 4 of sections \ref{sec:sims} and \ref{app:pfer_sims} in the supplement. \textit{Second}, the discussion following Proposition \ref{prop:general_n_r} reveals that when the participating studies are nearly perfectly correlated, $e-$Filter may not be valid for simultaneous error rate control unless $\kappa$ is chosen to be close to $1$, in which case it can be extremely conservative. Although such a scenario is unlikely in GWAS applications, it represents a limitation in other practical settings.
		
		While our work centers on replicability analysis, extending $e-$Filter to broader settings involving dependent studies presents an exciting direction for future research. Potential applications include detecting conditional associations consistent across heterogeneous environments \citep{10.1093/biomet/asab055} and accommodating scenarios where some studies report only binary decisions rather than $p$-values \citep{tang2014imputation}. Moreover, it will be valuable to explore whether the $e-$Filter and $e-$PCH procedures can be improved using the $e$-Closure Principle \citep{xu2025bringing}. We are actively investigating these extensions at the time of writing this article.
\bibliography{references}
\newpage
\phantomsection\label{supplementary-material}
\bigskip

\begin{center}

{\large\bf SUPPLEMENTARY MATERIAL: \\
 Multiple Testing of Partial Conjunction Hypotheses for Assessing Replicability Across Dependent Studies}

\end{center}
This supplement is organized as follows: Section \ref{app:proofs} presents the proofs of all theoretical statements in Section \ref{sec:theory} of the manuscript. Section \ref{app:lemmata1_3} provides the proofs of lemmata \ref{lem:cauchy_pvals}--\ref{lem:ef_equivalence} referenced in the main text. Section \ref{app: fig_details} provides additional information pertaining to figures \ref{fig:motfig1}, \ref{fig:fig_sec3_1}, \ref{fig:fig_sec3_2} and \ref{fig:fig_sec_4}. Numerical experiments for PFER control are reported in Section \ref{app:pfer_sims}. Section \ref{app:data_source} lists the data sources for the real data application and Section \ref{app:reactome} provides technical details on the metrics reported in Table \ref{tab:pathway_1}.
\appendix

\section{Proofs of theoretical statements in Section \ref{sec:theory}}
\label{app:proofs}
\subsection{Proof of Proposition \ref{prop:general_n_r}}
Let $x=\varphi^{-1}(1/\gamma)$. We have $\kappa \geq \kappa^{\star} \geq d_2-d_1+1$. This means $d_1 \geq d_2+1-\kappa$. Thus, 
\begin{align*}
	& x^{d_1} \leq x^{d_2+1-\kappa}\\
	\implies & x^{d_1} \leq \frac{x^{1-\kappa}}{{\kappa}} \cdot x^{d_2} \\
	\implies & x^{d_1} \leq \gamma \cdot x^{d_2}\\
	\implies & \mathrm{pr}\left[\varphi\left(S_j\right)>1 / \gamma\right]  \leq \gamma \cdot \mathrm{pr}\left[\varphi\left(F_j\right)>1 / \gamma\right].
\end{align*}




\subsection{Proof of \autoref{thm:e_adafilter_fwer_control}}
For each $j=1,2, \ldots, m$, we define
$$
\gamma_j=\sup \left[\gamma \in[0, \alpha]: \gamma\left\{1+\sum_{s \neq j}^m \mathbb I(\varphi(F_s)>1/\gamma)\right\} \leq \alpha\right]
$$
which is independent from ($F_j, S_j$) because $\left(P_{ij}\right)_{n \times m}$ contain columnwise independent valid $p$-values. It is obvious from the definition that we always have $\gamma_j \leq \gamma_e^{\mathrm{PFER}}$. Specifically, if $\varphi(F_j)>1/\gamma_e^{\mathrm{PFER}}$, then $\gamma_j=\gamma_e^{\mathrm{PFER}}$. Thus, as $F_j \leq S_j$ always holds, the PFER is
$$
\begin{aligned}
	E(V) &=E\left(\sum_{j=1}^m \mathbb I(\varphi(S_j)>1/\gamma_e^{\mathrm{PFER}}) \cdot \mathbb{I}(v_j=0)\right) \\& =E\left(\sum_{j=1}^m \mathbb{I}(\varphi(S_j)>1/\gamma_0^{\mathrm{Bon}}) \mathbb{I}(\varphi(F_j)>\gamma_e^{\mathrm{PFER}}) \cdot \mathbb{I}(v_j=0)\right) \\
	& =\sum_{j=1}^m E\left(\mathbb{I}(\varphi(S_j)>1/\gamma_j) \cdot \mathbb{I}(\varphi(F_j)>1/\gamma_e^{\mathrm{PFER}}) \cdot \mathbb{I}(v_j=0)\right) \\
	& =\sum_{j=1}^m E\left(\mathbb{I}(\varphi(S_j)>1/\gamma_j) \cdot \mathbb{I}(\varphi(F_j)>1/\gamma_j)\right) \cdot \mathbb{I}(v_j=0)
\end{aligned}
$$
The last equality holds as both $\gamma_j \leq \gamma_e^{\mathrm{PFER}}$ and $F_j \leq S_j$ hold for any $j$.
Now using Proposition \ref{prop:general_n_r} and the fact that $\gamma_j \leq \gamma_e^{\mathrm{PFER}}$, we have
$$
\begin{aligned}
	E(V) & =\sum_{j=1}^m E\left(\mathbb{I}(\varphi(S_j)>1/\gamma_j) \cdot \mathbb{I}(\varphi(F_j)>1/\gamma_j)\right) \cdot \mathbb{I}(v_j=0) \\
	& =\sum_{j=1}^m E\left(E\left[\mathbb{I}(\varphi(S_j)>1/\gamma_j) \mid \gamma_j, \mathbb{I}(\varphi(F_j)>1/\gamma_j)\right] \cdot \mathbb{I}(\varphi(F_j)>1/\gamma_j) \cdot \mathbb{I}(v_j=0)\right)\\
	& \leq \sum_{j=1}^m E\left(\gamma_j \cdot \mathbb{I}(\varphi(F_j)>1/\gamma_j) \cdot \mathbb{I}(v_j=0)\right) \\
	& \leq E\left(\gamma_e^{\mathrm{PFER}} \cdot \sum_{j=1}^m \mathbb{I}(\varphi(F_j)>1/\gamma_e^{\mathrm{PFER})} \cdot \mathbb{I}(v_j=0)\right) \leq \alpha .
\end{aligned}
$$
The last inequality holds as $\gamma_e^{\mathrm{PFER}}$ itself satisfies $\gamma_e^{\mathrm{PFER}} \cdot \sum_{j=1}^m \mathbb{I}(\varphi(F_j)>1/\gamma_e^{\mathrm{PFER}}) \leq \alpha$.

\subsection{Proof of \autoref{thm:e_adafilter_fdr_control}} 

We shall need the following result to show FDR control our proposed method.

\setcounter{lemma}{3}
\begin{lemma}\label{ineq_on_gamma}
	Suppose, for each $j=1,2, \ldots, m$, we define
	$$\gamma_j=\sup \left\{\gamma \in[0, \alpha]: \gamma \cdot\left(1+\sum_{k \neq j}^m \mathbb{I}(\varphi(F_k)>1/\gamma)\right) \leq \alpha \cdot\left(1+\sum_{k \neq j}^m \mathbb{I}(\varphi(S_k)>1/\gamma)\right)\right\}.$$ 
	Then, $\gamma_j \geq \gamma_e^{\mathrm{FDR}}$. 
\end{lemma}

\begin{proof}[of Lemma \ref{ineq_on_gamma}]
	We have 
	$$\gamma_e^{\mathrm{FDR}}=\sup \left\{\gamma \in[0, \alpha] : \gamma \cdot t_1(\gamma)\leq \alpha\right\} \quad \text{and} \quad \gamma_{j}=\sup \left\{\gamma \in[0, \alpha] : \gamma \cdot t_2(\gamma)\leq \alpha\right\}$$
	where 
	$$t_1(\gamma):= \frac{\sum_{j=1}^m \mathbb{I}(\varphi(F_j)>1/\gamma)}{\sum_{j=1}^m \mathbb{I}(\varphi(S_j)>1/\gamma) \vee 1}, \quad t_2(\gamma):=\frac{1+ \sum_{k \neq j}^m \mathbb{I}(\varphi(F_k)>1/\gamma)}{1+\sum_{k \neq j}^m \mathbb{I}(\varphi(S_k)>1/\gamma)}.$$
	Hence, it is enough to show that $\gamma_e^{\mathrm{FDR}}\cdot t_{2}(\gamma_e^{\mathrm{FDR}})\leq \alpha$. We consider three cases: 
	
	\noindent \textbf{Case 1. $\varphi(S_{j})>1/\gamma_e^{\mathrm{FDR}}$ }
	
	In this case $\varphi(F_{j})>1/\gamma_e^{\mathrm{FDR}}$ as well, since $f$ is decreasing. Consequently, one has $t_{1}(\gamma_e^{\mathrm{FDR}})=t_2(\gamma_e^{\mathrm{FDR}})$. Indeed, for every $\gamma>\gamma_e^{\mathrm{FDR}}$, we have $t_{1}(\gamma)=t_2(\gamma)$. So, $\gamma_{j}=\gamma_e^{\mathrm{FDR}}$. 
	
	\noindent \textbf{Case 2. $\varphi(F_{j})> 1/\gamma_e^{\mathrm{FDR}} \geq \varphi(S_j)$ }
	
	Here we shall show that $t_{2}(\gamma_e^{\mathrm{FDR}})\leq t_{1}(\gamma_e^{\mathrm{FDR}})$. Let $a= \sum_{k \neq j}^m \mathbb{I}(\varphi(F_k)>1/\gamma_e^{\mathrm{FDR}})$ and $b = \sum_{k \neq j}^m \mathbb{I}(\varphi(S_k)>1/\gamma_e^{\mathrm{FDR}})$. Then, 
	\begin{align*}
		& t_{2}(\gamma_e^{\mathrm{FDR}}) \leq t_{1}(\gamma_e^{\mathrm{FDR}}) \\
		\iff & \frac{1+a}{1+b} \leq \frac{a+\mathbb{I}(\varphi(F_j)>1/\gamma_e^{\mathrm{FDR}})}{\{b+\mathbb{I}(\varphi(S_j)>1/\gamma_e^{\mathrm{FDR}}))\}\vee 1}.
	\end{align*}
	When, $\varphi(F_{j})> 1/\gamma_e^{\mathrm{FDR}} \geq \varphi(S_j)$ then the numerators match on two sides whereas the RHS has smaller denominator.
	
	\noindent \textbf{Case 3. $\varphi(F_j)\leq 1/\gamma_e^{\mathrm{FDR}}$}
	
	The definition of $\gamma_e^{\mathrm{FDR}}$ ensures that $\gamma_e^{\mathrm{FDR}} \leq \alpha$. Also, when $\varphi(F_j)\leq 1/\gamma_e^{\mathrm{FDR}}$, we have $\gamma_e^{\mathrm{FDR}} \cdot \frac{a}{b} \leq \alpha.$ So, 
	\begin{align*}
		& \gamma_e^{\mathrm{FDR}}\cdot (1+a) \leq \alpha \cdot(1 + b) \\
		\implies  & \gamma_e^{\mathrm{FDR}}\cdot t_{2}(\gamma_e^{\mathrm{FDR}})\leq \alpha.
	\end{align*}
\end{proof}
We note that $\gamma_j$ is independent of $(F_j, S_j)$. In addition, the definition of $\gamma_e^{\mathrm{FDR}}$ ensures that
$$
\frac{\gamma_e^{\mathrm{FDR}} \sum_{j=1}^m \mathbb{I}(\varphi(F_j)>1/\gamma_e^{\mathrm{FDR}})}{\sum_{j=1}^m \mathbb{I}(\varphi(S_j)>1/\gamma_e^{\mathrm{FDR}}) \vee 1} \leq \alpha.
$$
This results in the following expression of the FDR of e-Filter: 
$$
\begin{aligned}
	E\left(\frac{V}{R \vee 1}\right)
	& =E\left\{\frac{\sum_{j=1}^m \mathbb{I}(\varphi(S_j)>1/\gamma_e^{\mathrm{FDR}})  \cdot \mathbb{I}(v_j=0)}{\sum_{j=1}^m \mathbb{I}(\varphi(S_j)>1/\gamma_e^{\mathrm{FDR}})  \vee 1}\right\} \\
	& \leq \alpha \cdot E\left[\frac{\sum_{j=1}^m \mathbb{I}(\varphi(S_j)>1/\gamma_e^{\mathrm{FDR}})  \cdot \mathbb{I}(v_j=0)}{\left\{\gamma_e^{\mathrm{FDR}} \sum_{j=1}^m \mathbb{I}(\varphi(F_j)>1/\gamma_e^{\mathrm{FDR}})\right\} \vee \alpha}\right] \\
	& =\alpha \cdot \sum_{j=1}^m E\left[\frac{\mathbb{I}(\varphi(S_j)>1/\gamma_e^{\mathrm{FDR}})  \cdot \mathbb{I}(v_j=0)}{ \left\{\gamma_e^{\mathrm{FDR}} \sum_{k=1}^m \mathbb{I}(\varphi(F_k)>1/\gamma_e^{\mathrm{FDR}})\right\} \vee \alpha }\right] \\
	& = \alpha \cdot \sum_{j=1}^m E \left[ \frac{\mathbb{I}(\varphi(S_j)>1/\gamma_e^{\mathrm{FDR}}) \cdot \mathbb{I}(\varphi(S_j)>1/\gamma_j) \cdot \mathbb{I}(v_j=0)}{\left\{\gamma_j\left(1+\sum_{k \neq j}^m \mathbb{I}(\varphi(F_k)>1/\gamma_j)\right)\right\} \vee \alpha}\right] \quad (\text{utilizing Lemma \ref{ineq_on_gamma}})\\
	& \leq \alpha \cdot \sum_{j=1}^m E \left[
	\frac{\mathbb{I}(\varphi(S_j)>1/\gamma_j) \cdot \mathbb{I}(v_j=0)}{\left\{\gamma_j\left(1+\sum_{k \neq j}^m \mathbb{I}(\varphi(F_k)>1/\gamma_j)\right)\right\} \vee \alpha}\right].
\end{aligned}
$$

\noindent Similarly to \cite{Owen}, let $P_{\cdot(-j)}$ contain all $P_{\cdot k}$ for $k \neq j$. Hence $P_{\cdot(-j)}$ determines $\gamma_j$ and all $\varphi(F_k)$ for $k \neq j$. 

\noindent From the inequality obtained earlier, we have
\begin{align*}
	E\left(\frac{V}{R \vee 1}\right) & \leq \alpha \cdot \sum_{j=1}^m E\left(E\left[\left.\frac{\mathbb{I}(\varphi(S_j)>1/\gamma_j) \cdot \mathbb{I}(v_j=0)}{\left[\gamma_j\left(1+\sum_{k \neq j}^m \mathbb{I}(\varphi(F_k)>1/\gamma_j)\right)\right] \vee \alpha} \right\rvert\, P_{\cdot(-j)}\right]\right) \\
	& = \alpha \cdot \sum_{j=1}^m E\left(\frac{E\left[\mathbb{I}(\varphi(S_j)>1/\gamma_j) \mid P_{\cdot(-j)}\right]}{\left[\gamma_j\left(1+\sum_{k \neq j}^m \mathbb{I}(\varphi(F_k)>1/\gamma_j\right)\right] \vee \alpha}\right) \quad \text{(since $\gamma_j$ is independent of $(F_j, S_j)$)}\\
	& \leq \alpha \cdot \sum_{j=1}^m E\left(\frac{\gamma_j E\left[\mathbb{I}(\varphi(F_j)>1/\gamma_j) \mid P_{\cdot(-j)}\right]}{\left[\gamma_j\left(1+\sum_{k \neq j}^m \mathbb{I}(\varphi(F_k)>1/\gamma_j)\right)\right] \vee \alpha}\right) \quad \text{(from Proposition \ref{prop:general_n_r})} \\
	& =\alpha \cdot \sum_{j=1}^m E\left(\frac{\gamma_j \mathbb{I}(\varphi(F_j)>1/\gamma_j)}{\left[\gamma_j\left(1+\sum_{k \neq j}^m \mathbb{I}(\varphi(F_k)>1/\gamma_j)\right)\right] \vee \alpha}\right) \\
	&=\alpha \cdot \sum_{j=1}^m E\left(\frac{\gamma_j \mathbb{I}(\varphi(F_j)>1/\gamma_j)}{\left[\gamma_j\left(\sum_{k=1}^M \mathbb{I}(\varphi(F_k)>1/\gamma_j)\right)\right] \vee \alpha}\right)  \\
	&=\alpha \cdot \sum_{j=1}^m E\left(\frac{ \mathbb{I}(\varphi(F_j)>1/\gamma_j)}{\left[\left(\sum_{k=1}^m \mathbb{I}(\varphi(F_k)>1/\gamma_j)\right)\right] \vee 1}\right). 
\end{align*}

\noindent Without any loss of generality, we may assume $F_1 \leq F_2 \leq \cdots \leq F_m$. One always has
$$\frac{\mathbb{I}(\varphi(F_j)>1/\gamma_j)}{\sum_{k=1}^m \mathbb{I}(\varphi(F_k)>1/\gamma_j) \vee 1} \leq \frac{\mathbb{I}(\varphi(F_j)>1/\gamma_j)}{\sum_{k=1}^m \mathbb{I}(\varphi(F_k) \geq \varphi(F_j))} \leq \frac{\mathbb{I}(\varphi(F_j)>1/\gamma_j)}{\sum_{k=1}^m \mathbb{I}(F_k \leq F_j)} \leq \frac{\mathbb{I}(\varphi(F_j)>1/\gamma_j)}{j} .$$
\noindent Thus,
\begin{align*}
	E\left(\frac{V}{R \vee 1}\right) \leq \alpha \cdot E\left(\mathbb{I}(\varphi(F_j)>1/\gamma_j)\right)\cdot\sum_{j=1}^{m} \frac{1}{j}
	& = \alpha \cdot E\left(\mathbb{I}(F_j<\varphi^{-1}(1/\gamma_j))\right)\cdot\sum_{j=1}^{m} \frac{1}{j} \\
	& = \alpha \cdot \left(\varphi^{-1}(1/\gamma_j)\right)^{d_1}\cdot\sum_{j=1}^{m} \frac{1}{j} \\
	& \leq \alpha \cdot {\gamma_j}^{d_1} \cdot\sum_{j=1}^{m} \frac{1}{j} \\
	& \leq \alpha^{1+d_1}\cdot \sum_{j=1}^{m} \frac{1}{j}.
\end{align*}

\subsection{Proof of \autoref{thm:asymp_e_filter_fdr_control}}

\noindent \cite{Owen} define four ECDFs
$$
\begin{aligned}
	F_{0, m}(\gamma) & :=\frac{1}{m_0} \sum_{j \in \mathcal{H}_0^{r / n}} \mathbb{I}(F_j<\gamma), & F_{1, m}(\gamma) & :=\frac{1}{m_1} \sum_{j \in \mathcal{H}_1^{r / n}} \mathbb{I}(F_j<\gamma) \\
	S_{0, m}(\gamma) & :=\frac{1}{m_0} \sum_{j \in \mathcal{H}_0^{r / n}} \mathbb{I}(S_j<\gamma), \quad \text { and } & S_{1, m}(\gamma) & :=\frac{1}{m_1} \sum_{j \in \mathcal{H}_1^{r / n}} \mathbb{I}(S_j<\gamma) .
\end{aligned}
$$
They show that $F_{0, m}(\gamma) \xrightarrow{p} \tilde{F}_0(\gamma)$, $F_{1, m} \xrightarrow{p} \tilde{F}_1, S_{0, m} \xrightarrow{p} \tilde{S}_0$ and $S_{1, m} \xrightarrow{p} \tilde{S}_1$ uniformly in $0 \leq \gamma \leq 1$ under Assumptions \ref{assume:2} and \ref{asymp_assump_2}. 

\noindent We define
$$
F_m(\gamma):=\frac{1}{m} \sum_{i=1}^m \mathbb{I}(\left(F_i<\gamma\right)), \quad S_m(\gamma):=\frac{1}{m} \sum_{i=1}^m \mathbb{I}(\left(S_i<\gamma\right)) \quad \text { and } \quad f_m(\gamma):=\frac{\gamma F_m(\varphi^{-1}(1/\gamma))}{S_m(\varphi^{-1}(1/\gamma)) \vee \frac{1}{m}}.$$
A direct conclusion from the above-mentioned convergence of ECDFs is that 
$$
F_m(\gamma) \xrightarrow{p} \tilde{F}(\gamma), \quad S_m(\gamma) \xrightarrow{p} \tilde{S}(\gamma) \quad \text { and } \quad f_m(\gamma) \xrightarrow{p} f_e^{\infty}(\gamma)
$$
all hold uniformly in $\gamma \in[0, \alpha]$.

\noindent As a consequence, for $\forall x \in[0, \alpha]$,

$$
\begin{aligned}
	\inf _{\gamma \in[x, \alpha]} f_m(\gamma) & \geq \inf _{\gamma \in[x, \alpha]}\left[f_m(\gamma)-f_e^{\infty}(\gamma)\right]+\inf _{\gamma \in[x, \alpha]} f_e^{\infty}(\gamma) \\
	& \geq-\sup _{\gamma \in[x, \alpha]}\left|f_m(\gamma)-f_e^{\infty}(\gamma)\right|+\inf _{\gamma \in[x, \alpha]} f_e^{\infty}(\gamma) \\
	& \xrightarrow{p} \inf _{\gamma \in[x, \alpha]} f_e^{\infty}(\gamma)
\end{aligned}
$$
\noindent In addition, for any $\epsilon>0$ we have
$$\begin{aligned}
	& \limsup _{m \rightarrow \infty} \mathrm{pr}\left(\inf _{\gamma \in[x, \alpha]} f_m(\gamma) \leq \alpha\right) \\
	\leq & \mathrm{pr}\left(\liminf _{m \rightarrow \infty}\left(\inf _{\gamma \in[x, \alpha]} f_m(\gamma)\right)-\inf _{\gamma \in[x, \alpha]} f_e^{\infty}(\gamma)<\epsilon\right)+\mathbb{I}( \left \{ \displaystyle \inf _{\gamma \in[x, \alpha]} f_e^{\infty}(\gamma) \leq \alpha+\epsilon \right \})
\end{aligned}
$$
Letting $\epsilon \rightarrow 0$, we have

$$
\limsup _{m \rightarrow \infty} \mathrm{pr}\left(\inf _{\gamma \in[x, \alpha]} f_m(\gamma) \leq \alpha\right) \leq \mathbb{I}( \left \{ \displaystyle \inf _{\gamma \in[x, \alpha]} f_e^{\infty}(\gamma) \leq \alpha+\epsilon \right \}).
$$
Similarly,
$$
\begin{aligned}
	\inf _{\gamma \in[x, \alpha]} f_m(\gamma) & \leq \sup _{\gamma \in[x, \alpha]}\left|f_m(\gamma)-f_e^{\infty}(\gamma)\right|+\inf _{\gamma \in[x, \alpha]} f_e^{\infty}(\gamma) \\
	& \xrightarrow{p} \inf _{\gamma \in[x, \alpha]} f_e^{\infty}(\gamma).
\end{aligned}
$$
We also have
$$
\liminf _{m \rightarrow \infty} \mathrm{pr}\left(\inf _{\gamma \in[x, \alpha]} f_m(\gamma) \leq \alpha\right) \geq \lim _{\epsilon \rightarrow 0}  \mathbb{I}( \left \{ \displaystyle \inf _{\gamma \in[x, \alpha]} f_e^{\infty}(\gamma) \leq \alpha-\epsilon \right \})= \mathbb{I}( \left \{ \displaystyle \inf _{\gamma \in[x, \alpha]} f_e^{\infty}(\gamma) \leq \alpha \right \}).
$$
\noindent Notice that 
$$\begin{aligned}
	E\left[\left(\gamma_e^{\mathrm{FDR}}\right)^k\right] & =\int_0^\alpha k x^{k-1} \mathrm{pr}\left(\gamma_e^{\mathrm{FDR}} \geq x\right) d x \\
	& =\int_0^\alpha k x^{k-1} \mathrm{pr}\left(\inf _{\gamma \in[x, \alpha]} f_m(\gamma) \leq \alpha\right) d x.
\end{aligned}
$$
Thus, by Fatou's lemma and earlier derivations, 
\begin{align*}
	\int_0^\alpha k x^{k-1} \mathbb{I}(\inf _{\gamma \in[x, \alpha]} f_e^{\infty}(\gamma)<\alpha) d x & \leq \liminf _{m \rightarrow \infty} E\left[\left(\gamma_e^{\mathrm{FDR}}\right)^k\right] \\
	& \leq \limsup _{m \rightarrow \infty} E\left[\left(\gamma_e^{\mathrm{FDR}}\right)^k\right] \\
	& \leq \int_0^\alpha k x^{k-1} \mathbb{I}(\inf _{\gamma \in[x, \alpha]} f_e^{\infty}(\gamma) \leq \alpha) d x.
\end{align*}
\noindent In addition, we have $\tilde{F}(\gamma) \geq \tilde{S}(\gamma)$ as $F_j \leq S_j$ for any hypothesis $j$, thus $\gamma_e^{\infty} \leq \alpha$. This combined with part 1 of Assumption 4 and $f_{e}^{\infty}$ 's left continuity also guarantee that
$$\left\{x: \gamma_e^{\infty} \geq x \text { and } x \leq \alpha\right\}=\left\{x: \inf _{\gamma \in[x, \alpha]} f_e^{\infty}(\gamma) \leq \alpha\right\}=\left\{x: \inf _{\gamma \in[x, \alpha]} f_e^{\infty}(\gamma)<\alpha\right\} \cup\left\{\gamma_e^{\infty}\right\}.$$
Hence,
$$\begin{aligned}
	\int_0^\alpha k x^{k-1} \mathbb{I}(\inf _{\gamma \in[x, \alpha]} f_e^{\infty}(\gamma)<\alpha) d x & =\int_0^\alpha k x^{k-1} \mathbb{I}(\inf _{\gamma \in[x, \alpha]} f_e^{\infty}(\gamma) \leq \alpha) d x \\
	& =\int_0^{\gamma_e^{\infty}} k x^{k-1} d x \\
	& =\left(\gamma_e^{\infty}\right)^k.
\end{aligned}$$ 
This implies
$$E[\gamma_e^{\mathrm{FDR}} - \gamma_e^\infty]^2 = E[(\gamma_e^{\mathrm{FDR}})^2] - 2\gamma_e^\infty \cdot E[\gamma_e^{\mathrm{FDR}}] + (\gamma_e^\infty)^2 \to 0.$$
Markov's inequality gives
$$\mathrm{pr}(|\gamma_e^{\mathrm{FDR}} - \gamma_e^\infty| \geq \epsilon) \leq \frac{E|\gamma_e^{\mathrm{FDR}} - \gamma_e^\infty|^2}{\epsilon^2} \to 0.$$
Hence, $\gamma_e^{\mathrm{FDR}} \xrightarrow{P} \gamma_e^\infty$, where $\gamma_e^\infty$ is a constant. The rest follows exactly in the same way as in \cite{Owen}.

\section{Proofs of Lemmata \ref{lem:cauchy_pvals}--\ref{lem:ef_equivalence}}
\label{app:lemmata1_3}
\subsection{Proof of Lemma \ref{lem:cauchy_pvals}}
\begin{proof}
	The proof follows by combining Theorem 1 of \cite{liu2020cauchy} with Lemma 1 of \cite{benjamini2008screening}. 
	
	In particular, Theorem 1 of \cite{liu2020cauchy} establishes that for a fixed $n$, $T(\bm P_j)=n^{-1}\sum_{i=1}^{n}\tan\{(0.5-P_{ij})\pi\}$ has an approximate Cauchy tail under arbitrary dependence between the underlying $n$ test statistics $\bm X_j$, as long as $E(\bm X_j)=\bm 0$ and for any $1\le r < s\le n$, $(X_{rj}, X_{sj})$ follows a bivariate normal distribution. Now, for a generic PC null $H_{0}^{r/n}$, let $(U_1,\ldots,U_{n-r+1},V_1,\ldots, V_{r-1})$ represent the $n$ dimensional $p-$value vector. Without loss of generality, suppose $U_1,\ldots,U_{n-r+1}$ are the $p-$values for which the base nulls are true, so that $\mathrm{pr}(U_i\le \gamma)\le \gamma$ for $i=1,\ldots,n-r+1$, and let $V_1,\ldots,V_{r-1}$ be the other $p-$values. For a combining function $f:[0,1]^n\to [0,1]$, denote $P_{r/n}=f(U_1,\ldots,U_{n-r+1},V_1,\ldots, V_{r-1})$ as the combined PC $p-$value and let $P^*_{r/n}=f(U_1,\ldots,U_{n-r+1},h_1(V_1),\ldots, h_{r-1}(V_{r-1}))$ for functions $h_i(x)\le x,~i=1,\ldots,r-1$ under $H_0^{r/n}$.  Lemma 1 of \cite{benjamini2008screening} asserts that if $f$ is non-decreasing in all its components, in the sense that $f(U_1,\ldots,U_{n-r+1},h_1(V_1),\ldots, h_{r-1}(V_{r-1}))\le f(U_1,\ldots,U_{n-r+1},V_1,\ldots, V_{r-1})$,  then under $H_{0}^{r/n}$, $\mathrm{pr}(P_{r/n}\le \gamma)\le \mathrm{pr}(P^*_{r/n}\le \gamma)$.
	
	Now in our context, ignoring subscript $j$,
	$$P_{r / n}=\mathrm{pr}\left[W \geq(n-r+1)^{-1}\sum_{i=r}^n \tan\{(0.5-P_{(i)})\pi\}\right].
	$$This is an increasing function of $P_1,\ldots,P_n$ and so Lemma 1 of \cite{benjamini2008screening} is applicable. Furthermore, in this case $P^{*}_{r/n} = f(U_1,\ldots,U_{n-r+1},0,\ldots,0)\sim U(0,1)$ approximately, from Theorem 1 of \cite{liu2020cauchy} (since $T(\mathbf{P}_j)$ has an approximate cauchy tail under the assumptions of our Lemma \ref{lem:cauchy_pvals}.). Therefore, under $H_{0}^{r/n}$, $\mathrm{pr}(P_{r/n}\le \gamma)\le \mathrm{pr}(P^{*}_{r/n}\le \gamma)$ which is approximately $\gamma$. 
\end{proof}
\subsection{Proof of Lemma \ref{lem:e-PCH}}
\begin{proof}
	We note that
	$$e_j=\frac{1}{n-r+1}\sum_{i=1}^{n-r+1}e_{(i)j} \leq \frac{1}{n} \sum_{i=1}^n e_{ij}.$$
	
	
	
		
		In the following, the max operator denotes the maximum with respect to all the possible base null configurations in $\mathcal{H}_{0j}^{r/n}$. One has the following
		\begin{align*}
			\max  E\left(e_{j}\right) \leqslant & E\left(\max e_{j}\right).
		\end{align*}
		The worst case maximum occurs when the maximum number of false PC null hypotheses are there. Under $\mathcal{H}_{0j}^{r/n}$, at most $r-1$ of $n$ hypotheses are non-null. So, consider the case when exactly, $r-1$ of $n$ hypotheses are non-null.
		Again, in the worst case, $\left\{e_{\left(n-r+2\right)j}, \ldots ., e_{\left(n\right)j}\right\}$ may take the largest value in their domain or may even be infinity. Even when they are infinity,
		$$
		e_{j}=\frac{1}{n-r+1}\sum_{i=1}^{n-r+1}e_{(i)j}
		$$
		where these $e$-variables correspond to the true null hypotheses. Hence, $e_j$ is a valid $e$-value under $\mathcal{H}_{0j}^{r/n}$. The rest follows by applying the $e-$BH procedure on these PCH $e-$values.
	\end{proof}
	\subsection{Proof of Lemma \ref{lem:ef_equivalence}}
	\begin{proof}
		We first show that the rejection sets $\{j:S_j^{\sf e}>1/\gamma_e^{\sf PFER}\}$ and $\{j:e_{j}^{\sf PFER}>1/\alpha\}$ are equivalent. For any $S_{j_1}^{\sf e}\ge S_{j_2}^{\sf e}$, we have $m_{j_1}\le m_{j_2}$ and $e_{j_1}^{\sf PFER}\ge e_{j_2}^{\sf PFER}$. So, we need to show that for any $j$, $S_{(j)}^{\sf e}>1/\gamma_e^{\sf PFER}\Leftrightarrow e_{(j)}^{\sf PFER}>1/\alpha$.
		
		If $S_{(j)}^{e}>1/\gamma_e^{\sf PFER}$ then,
		$$\dfrac{S_{(j)}^{\sf e}}{m_{(j)}}>\dfrac{1}{\gamma_e^{\sf PFER}}\dfrac{1}{\sum_{h=1}^m\mathbb I\{\varphi(F_h)>1/\gamma_e^{\sf PFER}\}}\ge 1/\alpha.
		$$
		So, we have $e_{(j)}^{\sf PFER}={S_{(j)}^{\sf e}}/{m_{(j)}}>1/\alpha$. Now, suppose $e_{(j)}^{\sf PFER}>1/\alpha$. Then,
		$$
		\dfrac{S_{(j)}^e}{\sum_{h=1}^m\mathbb I\{\varphi(F_h)>S_{(j)}^e\}}\ge \dfrac{S_{(j)}^e}{\sum_{h=1}^m\mathbb I\{\varphi(F_h)\ge S_{(j)}^e\}}>1/\alpha.$$ So, by the definition of $\gamma_e^{\sf PFER}$ (Definition \ref{def:eadaBon}), $S_{(j)}^e\ge 1/\gamma_e^{\sf PFER}$. If $S_{(j)}^e= 1/\gamma_e^{\sf PFER}$ then $e_{(j)}^{\sf PFER}>1/\alpha$ implies
		$$\dfrac{1/\gamma_e^{\sf PFER}}{\sum_{h=1}^m\mathbb I\{\varphi(F_h)\ge S_{(j)}^e\}}>1/\alpha.
		$$ But,
		$$\dfrac{1}{\alpha}\ge \inf_{\gamma>\gamma_e^{\sf PFER}}\dfrac{1/\gamma}{\sum_{h=1}^m\mathbb I\{\varphi(F_h)>1/\gamma\}}=\dfrac{1/\gamma_e^{\sf PFER}}{\sum_{h=1}^m\mathbb I\{\varphi(F_h)>=1/\gamma_e^{\sf PFER}\}}=\dfrac{S_{(j)}^{\sf e}}{\sum_{h=1}^m\mathbb I\{\varphi(F_h)>=S_{(j)}^{\sf e}\}},
		$$ which is a contradiction to $e_{(j)}^{\sf PFER}>1/\alpha$. So, we have $S_{(j)}^{\sf e}>1/\gamma_e^{\sf PFER}$. 
		
		We will now show that the rejection sets $\{j:S_j^{\sf e}>1/\gamma_e^{\sf FDR}\}$ and $\{j:e_{j}^{\sf FDR}>1/\alpha\}$ are equivalent. Again, for any $S_{j_1}^e\ge S_{j_2}^e$, we have $e_{j_1}^{\sf FDR}\ge e_{j_2}^{\sf FDR}$. So we need to show that for any $j$, $S_{(j)}^{\sf e}>1/\gamma_e^{\sf FDR}\Leftrightarrow e_{(j)}^{\sf FDR}>1/\alpha$. 
		
		Let $k_0=\max\{k:S_{(k)}^e>1/\gamma_e^{\sf FDR}\}$. Then,
		\begin{equation}
			\label{eq:lem3_1}
			S_{(j)}^e>1/\gamma_e^{\sf FDR}\Leftrightarrow S_{(j)}^e\ge S_{(k_0)}^e.
		\end{equation}
		Since
		$$ e_{(k_0)}^{\sf FDR} = \dfrac{S_{(k_0)}^e k_0}{\sum_{h=1}^m \mathbb I\{\varphi(F_h)\ge S_{(k_0)}^e\}}>\dfrac{k_0/\gamma_e^{\sf FDR}}{\sum_{h=1}^m \mathbb I\{\varphi(F_h)> 1/\gamma_e^{\sf FDR}\}}=\dfrac{\sum_{h=1}^m \mathbb I\{S_h^e> 1/\gamma_e^{\sf FDR}\}}{\gamma_e^{\sf FDR}\sum_{h=1}^m \mathbb I\{\varphi(F_h)> 1/\gamma_e^{\sf FDR}\}}\ge\dfrac{1}{\alpha},
		$$ we have 
		\begin{equation}
			\label{eq:lem3_2}
			S_{(j)}^e\ge S_{(k_0)}^e\leftrightarrow e_{(j)}^{\sf FDR}\ge e_{(k_0)}^{\sf FDR}>1/\alpha.
		\end{equation}
		Now, suppose there exists a $j$ such that $e_{(j)}^{\sf FDR}>1/\alpha$ and $e_{(j)}^{\sf FDR}<e_{(k_0)}^{\sf FDR}$. Then, we have $S_{(j)}^e\le 1/\gamma_e^{\sf FDR}$ and there exists a $k\ge j$ such that 
		$$1/\alpha<e_{(j)}^{\sf FDR}=\dfrac{S_{(k)}^ek}{\sum_{h=1}^m\mathbb I\{\varphi(F_h)\ge S_{(k)}^e\}}.
		$$
		Then, there exists a $\tilde{\gamma}_e$ with $1/\tilde{\gamma}_e<S_{(k)}^e\le 1/\gamma_e^{\sf FDR}$ such that
		$$1/\alpha<\dfrac{k/\tilde{\gamma_e}}{\sum_{h=1}^m\mathbb I\{\varphi(F_h)>1/\tilde{\gamma}_e\}}\le \dfrac{\sum_{h=1}^m\mathbb I\{S_h^e>1/\tilde{\gamma}_e\}}{\tilde{\gamma_e}\sum_{h=1}^m\mathbb I\{\varphi(F_h)>1/\tilde{\gamma}_e\}}, 
		$$ thus contradicting the definition of $\gamma_e^{\sf FDR}$ (Definition \ref{def:eadaBH}). So, we get
		\begin{equation}
			\label{eq:lem3_3}
			e_{(j)}^{\sf FDR}>1/\alpha\Leftrightarrow e_{(j)}^{\sf FDR}\ge e_{(k_0)}^{\sf FDR}.
		\end{equation}
		Finally, from equations \eqref{eq:lem3_1}--\eqref{eq:lem3_3}, we have
		$$S_{(j)}^{\sf e}>1/\gamma_e^{\sf FDR}\Leftrightarrow e_{(j)}^{\sf FDR}>1/\alpha.
		$$
	\end{proof}
\section{Details for figures \ref{fig:motfig1}, \ref{fig:fig_sec3_1}, \ref{fig:fig_sec3_2} and \ref{fig:fig_sec_4}}
\label{app: fig_details}
\textbf{The data generation scheme underlying Figure \ref{fig:motfig1} -- }we set $m=10,000$, $n=5$ and $r=2$.  Note that for each $n$ there are $2^n$ combinations of the base hypotheses $\mathcal H_{0j}$. So, denote $\pi_{00}$ as the probability of the global null combination and $\pi_1$ as the probability of the combinations where $H_{0j}^{r/n}$ is false. All remaining combinations where $H_{0j}^{r/n}$ is true have equal probabilities that sum to $1-\pi_{00}-\pi_1$. We fix $\pi_1=0.01$, $\pi_{00}=0.98$ and sample the $m\times n$ matrix of test statistics $\bm X$ from a matrix Normal distribution with $m\times n$ mean matrix $\bm\mu=(\mu_{ij}:1\le i\le n,~1\le j\le m)$, $m\times m$ row covariance matrix $\bm \Sigma_1=(\Sigma_1)_{i,j}$ and $n\times n$ column covariance matrix $\bm \Sigma_2$. In this scenario, $\Sigma_{1,ij}=\mathbb I(i\ne j)$ and $\bm \Sigma_2=\rho\bm 1\bm 1^T+(1-\rho)\bm I_n$. We test $H_{0ij}:\mu_{ij}=0~vs~H_{1ij}:\mu_{ij}\ne 0$ and calculate $P_{ij}=2\Phi(-|X_{ij}|)$. When $H_{0ij}$ is false, we set $\mu_{ij}\in\{-5,-4,4,5\}$ with equal probability. Figure \ref{fig:motfig1} reports the average FDR and average \% gain in power over $500$ repetitions of this data generation scheme.
\\[1ex]
\textbf{The data generation scheme underlying figures \ref{fig:fig_sec3_1} and \ref{fig:fig_sec3_2} -- }we borrow the scheme from Figure \ref{fig:motfig1} but fix $n=2$ and set $\mu_{ij}\in\{-6,-5,-4,4,5,6\}$ with equal probability whenever $H_{0ij}$ is false.
\\[1ex]
\textbf{Calculations underlying Figure \ref{fig:fig_sec_4} -- }ignoring subscript $j$, we have $S=\max(P_{(1)},\ldots,P_{(4)})$ and $F=\min(P_{(1)},\ldots,P_{(4)})$ in this example. Let $z=\Phi^{-1}(1-0.5 x)$. Then,
$ pr(S<x) =  pr(|X_1|>=z,\ldots,|X_4|>=z)$ and $ pr(F<x)=1- pr(|X_1|<=z,\ldots,|X_4|<=z)$. To obtain $d_1$, we first consider an equispaced grid $\Lambda_x$ of size $50$ on $(0,1)$ and then obtain $d_1$ as a root of the equation $\min_{x\in\Lambda_x}\{x^{d_1}- pr(S<x)\}$. Similarly, we recover $d_2$ as a root of the equation $\min_{x\in\Lambda_x}\{ pr(F<x)-x^{d_2}\}$. These roots are obtained numerically using the ``uniroot'' function in R. Finally, $\kappa^{*}=\max(0,d_2^{*}-d_1^{*}+1)$ where $d_1^{*}$ and $d_2^{*}$ are the roots.  
\section{Numerical experiments: PFER control}
\label{app:pfer_sims}
We report the performances of the six methods from Section \ref{sec:sims} for PFER control at the nominal level $\alpha=1$. We retain the five simulation scenarios from Section \ref{sec:sims} and discuss our findings below. Apart from Scenario 5, the results for PFER control are qualitatively similar to those for FDR control in Section \ref{sec:sims}.
\begin{table}[!h]
\centering
\scalebox{0.8}{\begin{tabular}{lcccccccc}
\hline
& \multicolumn{2}{c}{$\rho = 0.2$} & \multicolumn{2}{c}{$\rho = 0.4$} & \multicolumn{2}{c}{$\rho = 0.6$} & \multicolumn{2}{c}{$\rho = 0.8$} \\
\hline
          Method      & PFER & Recall & PFER & Recall & PFER & Recall & PFER & Recall \\
\hline
AdaFilter     &  0.478    &   0.882   &   1.182   &   0.884    &  3.191   &   0.886   &  8.325    &   0.888   \\
BH-$P^{B}_{r/n}$      &  0.006    &   0.536   &    0.008  &    0.537   &  0.023   &  0.539    &   0.099   &  0.543    \\
BH-$P^{C}_{r/n}$       & 0.006    &   0.540   &   0.008   &    0.541   &   0.029  &   0.544   &  0.147    &   0.548    \\
\hline
$e-$PCH   &   $<0.001$   &   0.288    &  $<0.001$    &    0.288   &   $<0.001$    &   0.288   &   0.001    &  0.288     \\
$e-$Filter C   &    0.013  &   0.649    & 0.026     &    0.652   &   0.107   &   0.658   &  0.440   &   0.666 \\
$e-$Filter B   &    0.015 &   0.649    & 0.026    &    0.653   &   0.088   &   0.658    &  0.310    &   0.666 
\\
\hline
\end{tabular}}
\caption{Scenario 1: Average PFER and Recall for different methods targeting a nominal PFER of $\alpha=1$. Results for each $n$ and $r$ are presented in Figure \ref{fig:scenario_1_pfer_rbyn} of the supplement.}
\label{tab:setting_1_pfer}
\end{table}
\\[1ex]
\noindent\textbf{Scenario 1 (Equicorrelated studies) - }table \ref{tab:setting_1_pfer} reports the average PFER and Recall across $B=500$ repetitions of the data generation scheme and six $(n,r)$ combinations. Results for each $n$ and $r$ are presented in Figure \ref{fig:scenario_1_pfer_rbyn} of the supplement. With the exception of AdaFilter, all methods control PFER at the nominal level for all values of $\rho$. The two variants of $e-$Filter have similar power and $e-$PCH is the most conservative procedure in this setting. Figure \ref{fig:scenario_1_pfer_rbyn} reveals that the gain in power of $e-$Filter is more substantial when $n=8$. 
\\[1ex]
\begin{figure}[!h]
\centering
\includegraphics[width=0.85\linewidth]{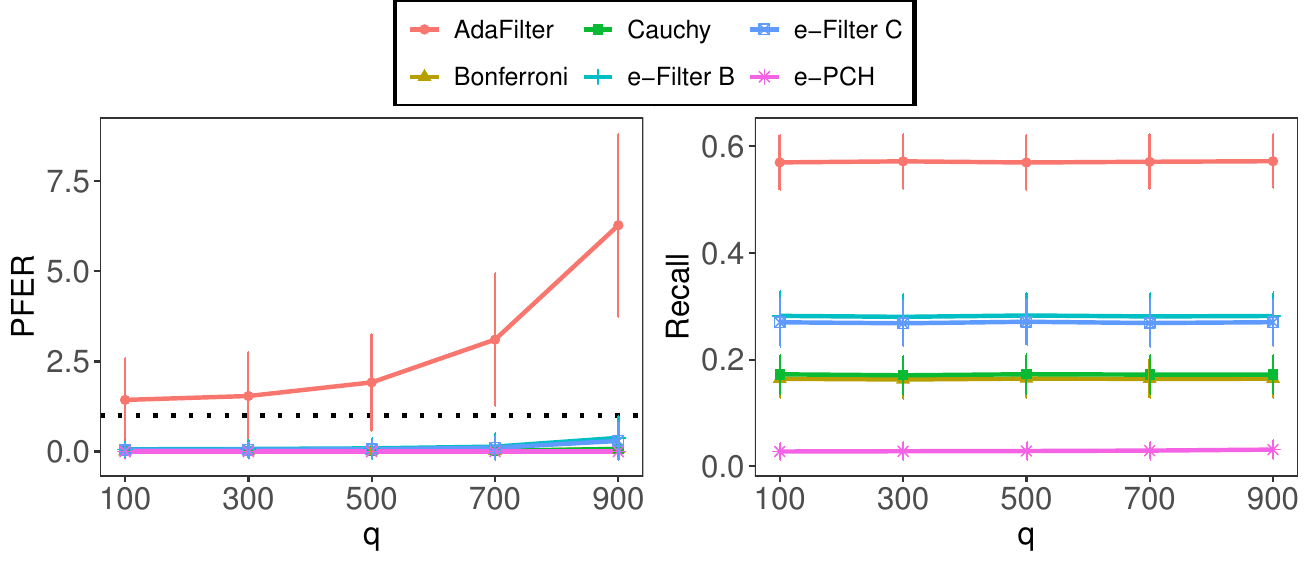}
\caption{Scenario 2: Average PFER and Recall as $q$ varies. Here $\alpha=1$. The error bars represent one standard deviation above and below the average PFER (Recall) from $B$ repetitions.}
\label{fig:scenario_3_pfer}
\end{figure}
\noindent\textbf{Scenario 2 (Common subjects) - }the results are presented in Figure \ref{fig:scenario_3_pfer}. AdaFilter fails to control PFER at the nominal level for all values of $q$. While all other methods control FDR, the two $e-$Filter variants exhibit the highest power.
\\[1ex]
\noindent\textbf{Scenario 3 (Negative study dependence) - }table \ref{tab:setting_2_pfer} and Figure \ref{fig:scenario_2_pfer_rbyn} report the results. AdaFilter controls PFER at the nominal level for all but the smallest value of $\rho$. Moreover, for $\rho\in\{-0.6,-0.8\}$, Figure \ref{fig:scenario_2_pfer_rbyn} reveals that it fails for the $2/4$ and $4/4$ configurations. All other methods control PFER at the nominal level. Moreover, apart from the $2/4$ and $4/4$ configurations, AdaFilter is generally more powerful than other methods and is only marginally better than $e-$Filter.
\begin{table}[!h]
\centering
\scalebox{0.8}{\begin{tabular}{lcccccccc}
\hline
& \multicolumn{2}{c}{$\rho = -0.2$} & \multicolumn{2}{c}{$\rho = -0.4$} & \multicolumn{2}{c}{$\rho = -0.6$} & \multicolumn{2}{c}{$\rho = -0.8$} \\
\hline
          Method      & PFER & Recall & PFER & Recall & PFER & Recall & PFER & Recall \\
\hline
AdaFilter     &  0.202    &   0.236   &   0.239   &   0.236    &  0.454   &   0.236  &  1.790    &   0.236   \\
BH-$P^{B}_{r/n}$      &  0.001    &   0.163   &    $<0.001$  &    0.163   &  $<0.001$   &  0.163    &   0.011   &  0.162   \\
BH-$P^{C}_{r/n}$       & 0.001    &   0.105   &   $<0.001$   &    0.104   &   $<0.001$ &   0.103   &  0.002   &   0.102    \\
\hline
$e-$PCH   &   $<0.001$   &   0.093    &  $<0.001$    &    0.091   &   $<0.001$    &   0.090   &   $<0.001$    &  0.089     \\
$e-$Filter C   &    0.010  &   0.120    & 0.006     &    0.119   &   0.007   &   0.119  &  0.010   &   0.118 \\
$e-$Filter B   &    0.011 &   0.196    & 0.007    &    0.195   &   0.013   &   0.195    &  0.060    &   0.195
\\
\hline
\end{tabular}}
\caption{Scenario 3: Average PFER and Recall for different methods targeting a nominal PFER of $\alpha=1$. Results for each $n$ and $r$ are presented in Figure \ref{fig:scenario_2_pfer_rbyn} of the supplement.}
\label{tab:setting_2_pfer}
\end{table}
\begin{table}[h!]
\centering
\scalebox{0.9}{\begin{tabular}{lcccc}
\hline
& \multicolumn{2}{c}{$\pi_{00} = 0.98$} & \multicolumn{2}{c}{$\pi_{00} = 0.8$} \\
\hline
Method      & PFER & Recall & PFER & Recall \\
\hline
AdaFilter     &  2.465    &   0.471    &   0.993   &   0.325 \\
BH-$P^{B}_{r/n}$      &  0.018    &   0.104   &    0.059   &    0.104\\
BH-$P^{C}_{r/n}$       & 0.021   & 0.109  &   0.063   &  0.109  \\
\hline
$e-$PCH   &   $<0.001$   &   0.022    &  0.001    &    0.022  \\
$e-$Filter C   &    0.089  &   0.196    & 0.041    &    0.106     \\
$e-$Filter B   &    0.099  &   0.202    & 0.048     &    0.113     \\
\hline
\end{tabular}}
\caption{Scenario 4: Average PFER and Recall for different methods targeting a nominal PFER of $\alpha=1$. Results for each $n$ and $r$ are presented in Figure \ref{fig:scenario_4_pfer_rbyn} of the supplement.}
\label{tab:scenario_4_pfer}
\end{table}

\noindent\textbf{Scenario 4 (Common controls) - }Table \ref{tab:scenario_4_pfer} and Figure \ref{fig:scenario_4_pfer_rbyn} report the results. When $\pi_{00}=0.98$, the $e-$Filter variants are the most powerful methods that, unlike AdaFilter, also control the PFER at the nominal level. However, at $\pi_{00}=0.8$, the PC $p-$value methods exhibit comparable power as the $e-$Filter variants. In fact, the former dominates the latter in power when $n\ne 8$. This trend reversal in power is expected when the proportion of the global null configuration is relatively small. 
\begin{table}[!h]
\centering
\scalebox{0.8}{\begin{tabular}{lccccccccccc}
\hline
&\multicolumn{2}{c}{$\rho = 0$} &
\multicolumn{2}{c}{$\rho = 0.2$} & \multicolumn{2}{c}{$\rho = 0.4$} & \multicolumn{2}{c}{$\rho = 0.6$} & \multicolumn{2}{c}{$\rho = 0.8$} \\
\hline
Method   & PFER & Recall   & PFER & Recall & PFER & Recall & PFER & Recall & PFER & Recall \\
\hline
AdaFilter   &12.488 & 0.763 &  11.340    &   0.770    &   9.678   & 0.787      &   8.847   &   0.812    &  10.540  & 0.847   \\
BH-$P^{B}_{r/n}$   & 3.450& 0.490  &  2.826    &   0.492  &  1.761    & 0.497    &  0.884   &  0.506     &   0.370   & 0.522   \\
BH-$P^{C}_{r/n}$    &3.465 & 0.494  & 2.842    &   0.495    &  1.778   &  0.501    &   0.920   &  0.510     &    0.451  &  0.526 \\
\hline
$e-$PCH  &0.549 &0.284 &  0.450   &   0.283    &   0.240  &  0.283  & 0.084   & 0.284    &   0.018   & 0.285   \\
$e-$Filter C &2.464 &0.545  &  2.141 &   0.551    &  1.519   &  0.566     &  0.997    &  0.590    &   0.840   &  0.626  \\
$e-$Filter B & 2.864& 0.549 &  2.511  &   0.555   &  1.826    &   0.570    &    1.163  &   0.594   & 0.784     &  0.629
\\
\hline
\end{tabular}}
\caption{Scenario 5: Average PFER and Recall for different methods targeting a nominal PFER of $\alpha=1$. Results for each $n$ and $r$ are presented in Figure \ref{fig:scenario_5_pfer_rbyn} of the supplement.}
\label{tab:setting_5_pfer}
\end{table}

\noindent\textbf{Scenario 5 (Scale mixture of Normals) - }Table \ref{tab:setting_5_pfer} and Figure \ref{fig:scenario_5_pfer_rbyn} report the results. Overall, $e-$PCH is the only method that controls PFER at the nominal level for all values of $\rho$. From Table \ref{tab:setting_5_pfer}, both the PC $p-$value methods and the $e-$Filter variants are valid only when $\rho\in\{0.6,0.8\}$. Moreover, Figure \ref{fig:scenario_5_pfer_rbyn} shows that these methods exhibit a substantially high empirical PFER for the $4/4$ configuration. Thus, in contrast to the results of Scenario 5 in Section \ref{sec:sims}, the $e-$Filter variants are relatively less robust to the misspecification of the base null distribution as far as PFER control is concerned.
\section{Data sources}
\label{app:data_source}
Here we list the sources of the data described in Section \ref{sec:realdata}.
\begin{enumerate}
    \item Biobank Japan (BBJ): \url{https://gwas.mrcieu.ac.uk/datasets/bbj-a-31/}
    \item Within-family GWAS consortium (FGC): \url{ https://gwas.mrcieu.ac.uk/datasets/ieu-b-4846/}
    \item European study (EUR2010): \url{https://gwas.mrcieu.ac.uk/datasets/ebi-a-GCST000759/}
    \item UK Biobank (UKB): \url{https://pan.ukbb.broadinstitute.org/downloads}
    \item Trans-ancestry global lipids genetic consortium (GLGC): \url{https://csg.sph.umich.edu/willer/public/glgc-lipids2021/}
\end{enumerate}
\section{Reactome odds ratio and combined score calculations}
\label{app:reactome}
\textbf{Odds ratio -- }
Let $\mathcal{G}$ denote the background gene universe of size $N_1 = |\mathcal{G}|$, 
and let $L \subset \mathcal{G}$ represent the user-supplied input gene list of size $N_2 = |L|$.
For a given pathway $P\subset \mathcal G$ of size $K = |P|$, let $k = |L \cap P|$ denote 
the number of input genes that belong to the pathway. 

The enrichment of $P$ within $L$ is assessed using a $2\times2$ contingency table:

\[
\begin{array}{c|cc|c}
 & \text{In pathway } P & \text{Not in pathway } P & \text{Total} \\
\hline
\text{In input list } L & k & N_2 - k & N_2 \\
\text{Not in input list} & K - k & N_1 - K - N_2 + k & N_1 - N_2 \\
\hline
\text{Total} & K & N_1 - K & N_1
\end{array}
\]

The \textit{odds ratio} (OR) is defined as
$$
\text{OR} = \dfrac{k (N_1 - K - N_2 + k)}{(N_2 - k)(K - k)}.
$$
It quantifies the ratio of the odds that a gene in the input list belongs to the pathway 
relative to the odds that a gene outside the list belongs to the pathway. 
An odds ratio greater than one indicates over-representation of the pathway among 
input genes, while an odds ratio less than one indicates depletion.

Statistical significance of the enrichment is assessed using Fisher’s exact test, 
based on the hypergeometric probability of observing at least $k$ overlapping genes:
$$
p = \sum_{i=k}^{\min(N_2,K)} \dfrac{\binom{K}{i}\binom{N_1-K}{N_2-i}}{\binom{N_1}{N_2}}.
$$
\\[1ex]
\textbf{Combined score -- }
Enrichment tools, such as Enrichr \citep{kuleshov2016enrichr}, additionally compute a \textit{combined score} 
to integrate both the magnitude and statistical significance of enrichment.
The combined score is defined as
$$
\text{Combined Score} = z\log(p),
$$
where $p$ is the unadjusted Fisher’s exact test $p-$value, and $z$ is a 
$z$-score measuring the deviation of the observed rank of the gene set 
from its expected rank under random permutations:
\[
z = \frac{l - \mu_l}{\sigma_l},
\]
where $l$ denotes the observed rank, and $\mu_l$ and $\sigma_l$ are the mean 
and standard deviation of ranks under the null distribution.

The combined score thus weights the statistical significance ($\log(p)$) 
by the standardized deviation from expectation ($z$), producing a composite 
enrichment measure. High combined scores correspond to pathways that are both 
strongly enriched and show larger-than-expected overlap with the input gene list. 
\begin{figure}[!h]
\centering
\includegraphics[width=1.05\linewidth]{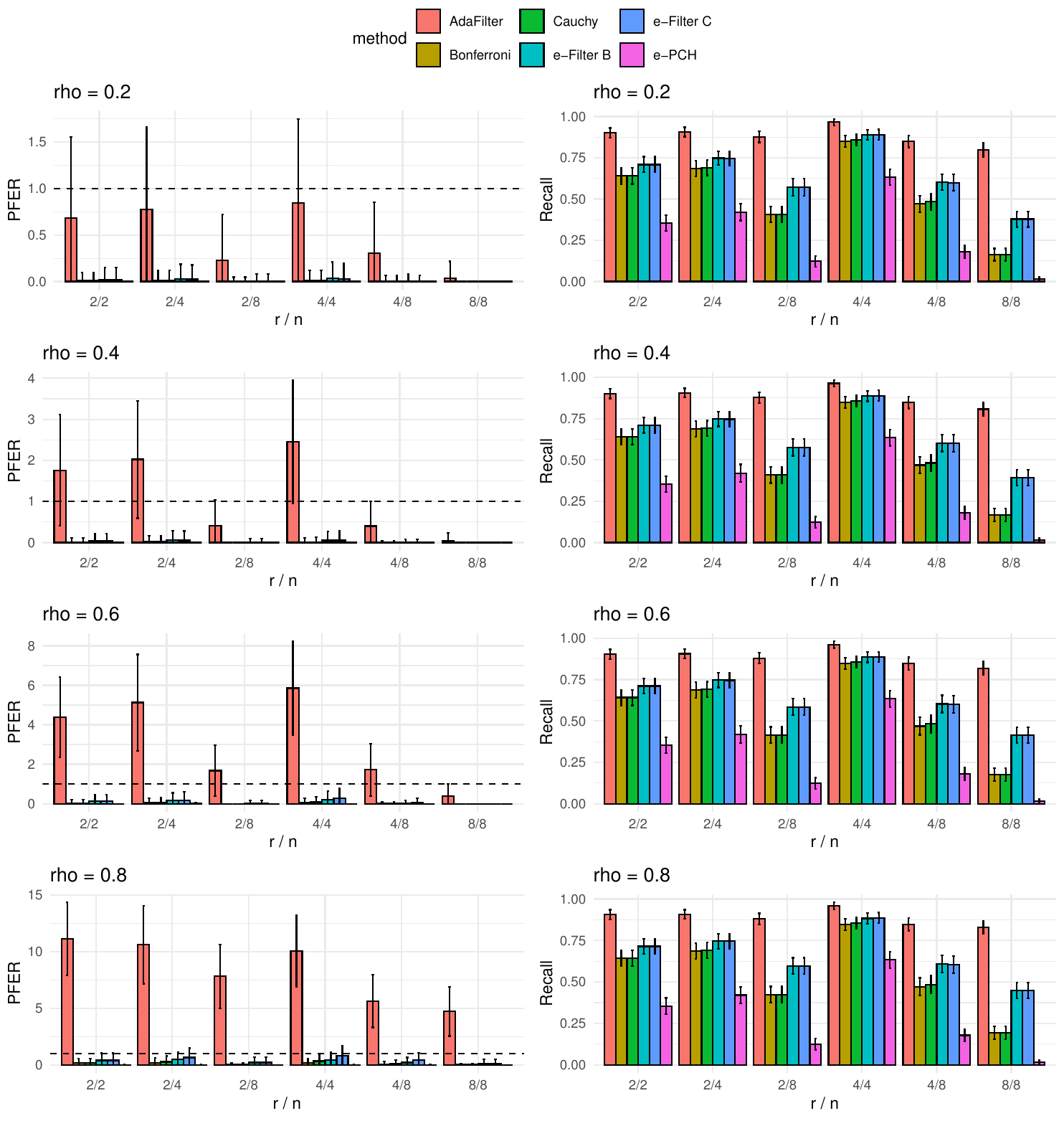}
\caption{Scenario 1: Average PFER and Recall for different methods across $B=500$ repetitions of the data generation scheme. Here $\alpha=1$. The error bars represent one standard deviation above and below the average PFER (Recall) from $B$ repetitions. Table \ref{tab:setting_1_pfer} provides the summary over the $r$ by $n$ combinations.
}
\label{fig:scenario_1_pfer_rbyn}
\end{figure}

\begin{figure}[!h]
\centering
\includegraphics[width=1.05\linewidth]{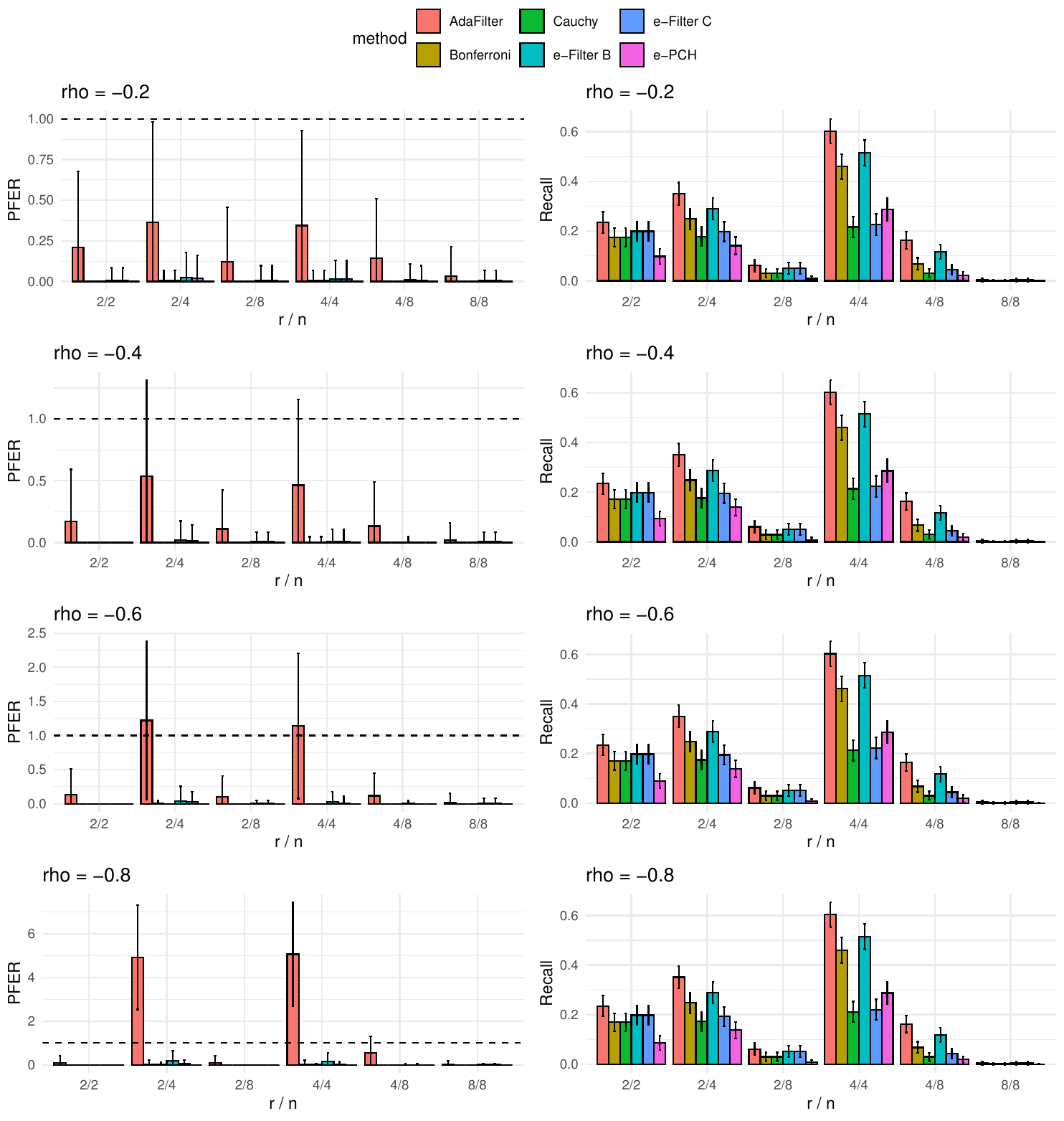}
\caption{Scenario 3: Average PFER and Recall for different methods across $B=500$ repetitions of the data generation scheme. Here $\alpha=1$. The error bars represent one standard deviation above and below the average PFER (Recall) from $B$ repetitions. Table \ref{tab:setting_2_pfer} provides the summary over the $r$ by $n$ combinations.
}
\label{fig:scenario_2_pfer_rbyn}
\end{figure}

\begin{figure}[!h]
\centering
\includegraphics[width=1\linewidth]{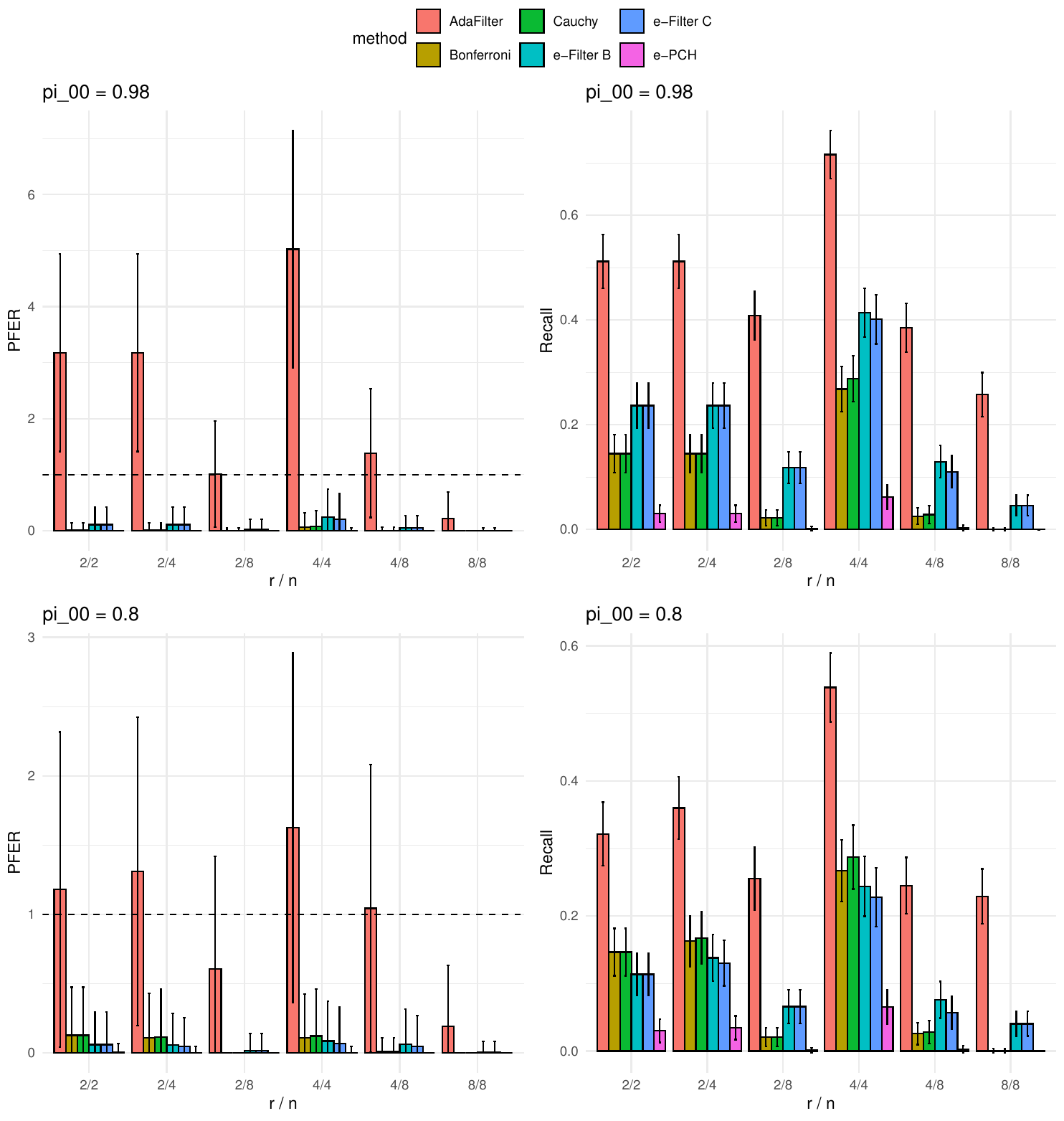}
\caption{Scenario 4: Average PFER and Recall for different methods across $B=500$ repetitions of the data generation scheme. Here $\alpha=1$. The error bars represent one standard deviation above and below the average PFER (Recall) from $B$ repetitions. Table \ref{tab:scenario_4_pfer} provides the summary over the $r$ by $n$ combinations.}
\label{fig:scenario_4_pfer_rbyn}
\end{figure}

\begin{figure}[!h]
\centering
\includegraphics[width=1.05\linewidth]{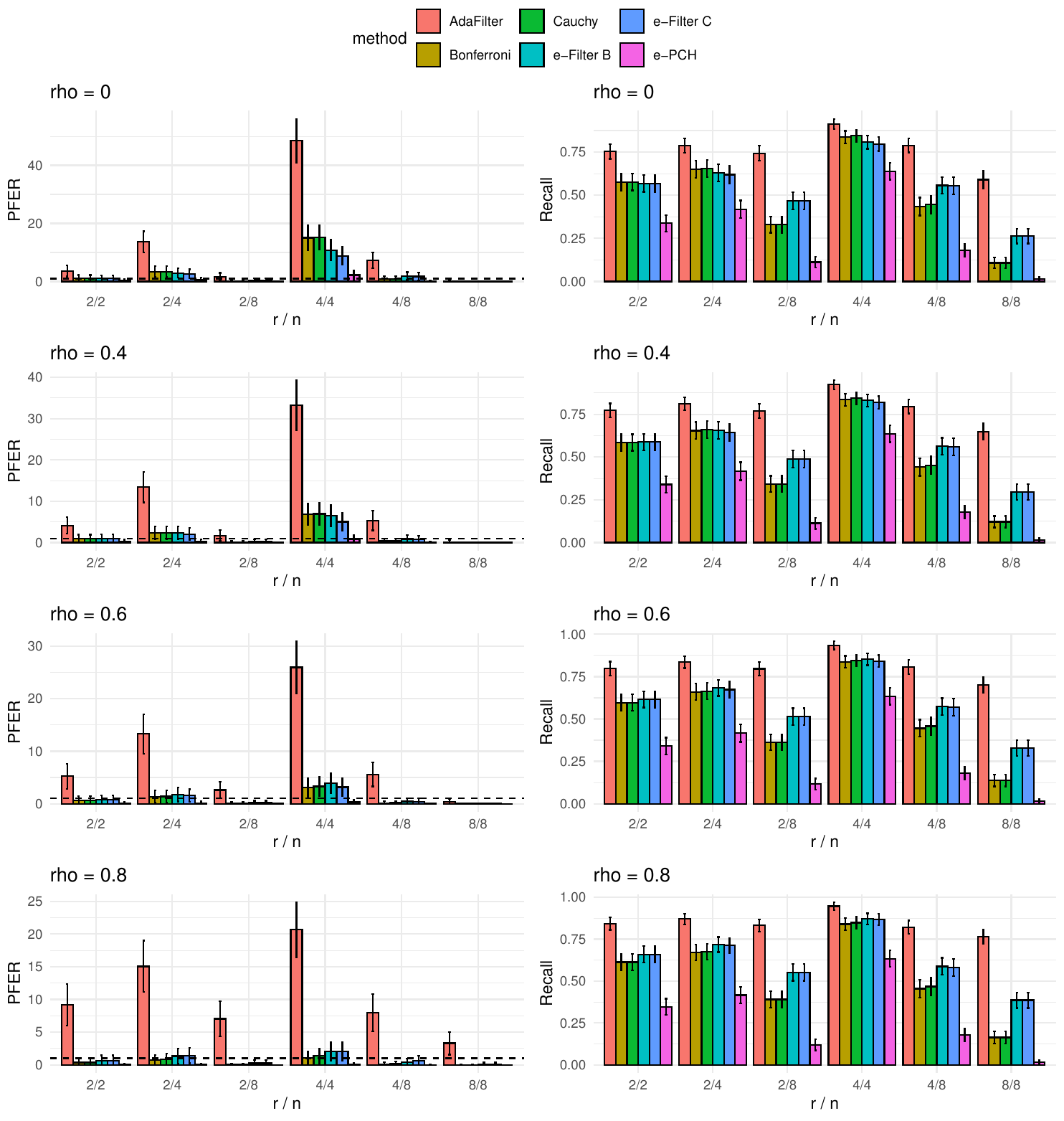}
\caption{Scenario 5: Average PFER and Recall for different methods across $B=500$ repetitions of the data generation scheme. Here $\alpha=1$. The error bars represent one standard deviation above and below the average PFER (Recall) from $B$ repetitions. Table \ref{tab:setting_5_pfer} provides the summary over the $r$ by $n$ combinations.}
\label{fig:scenario_5_pfer_rbyn}
\end{figure}

\begin{figure}[!h]
\centering
\includegraphics[width=1.05\linewidth]{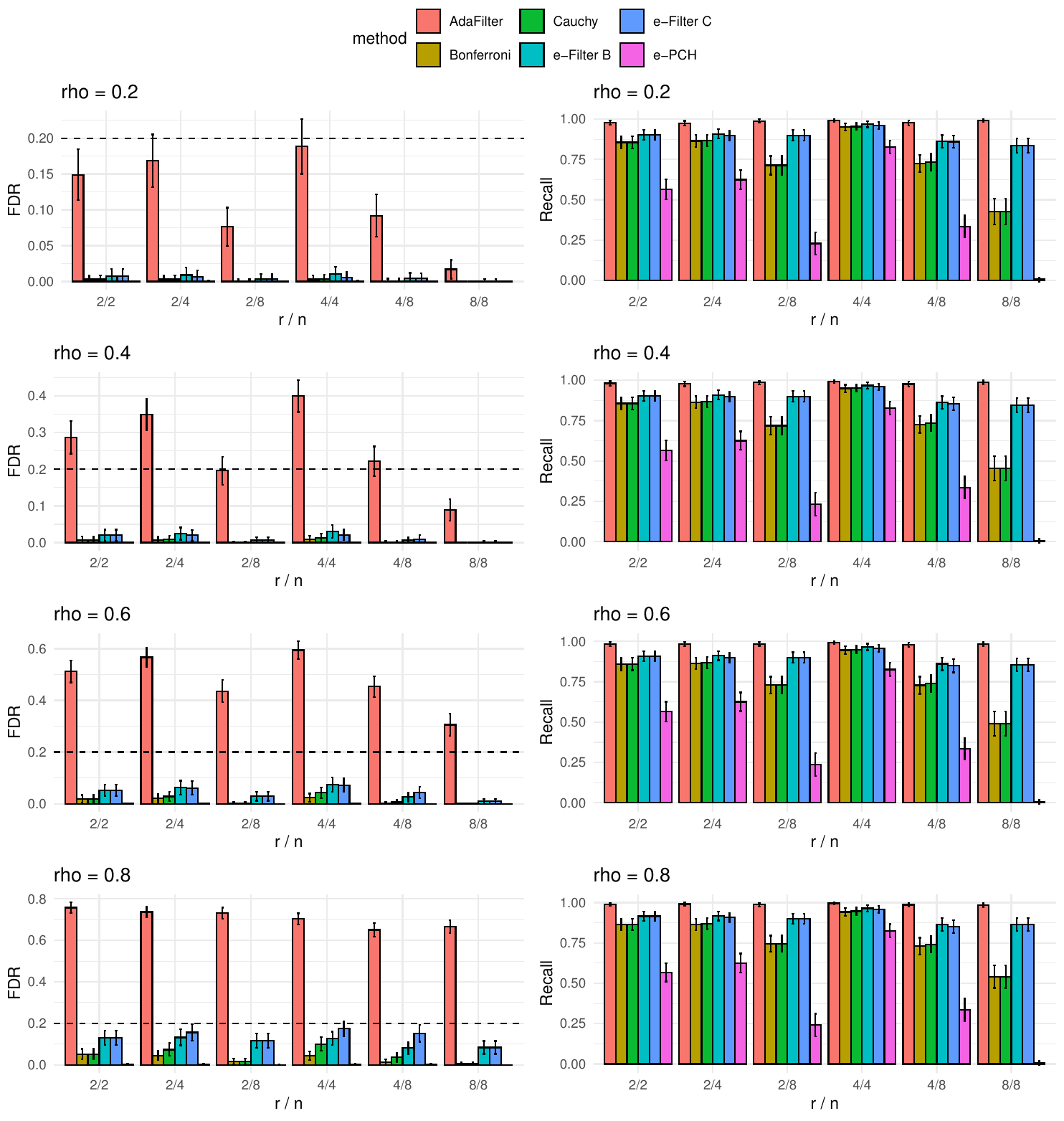}
\caption{Scenario 1: Average FDR and Recall of various methods across $B=500$ repetitions of the data generation scheme. Here $\alpha=0.2$. The error bars represent one standard deviation above and below the average FDR (Recall) from $B$ repetitions. Table \ref{tab:setting_1_fdr} provides the summary over the $r$ by $n$ combinations. 
}
\label{fig:scenario_1_fdr_rbyn}
\end{figure}
\begin{figure}[!h]
\centering
\includegraphics[width=1.05\linewidth]{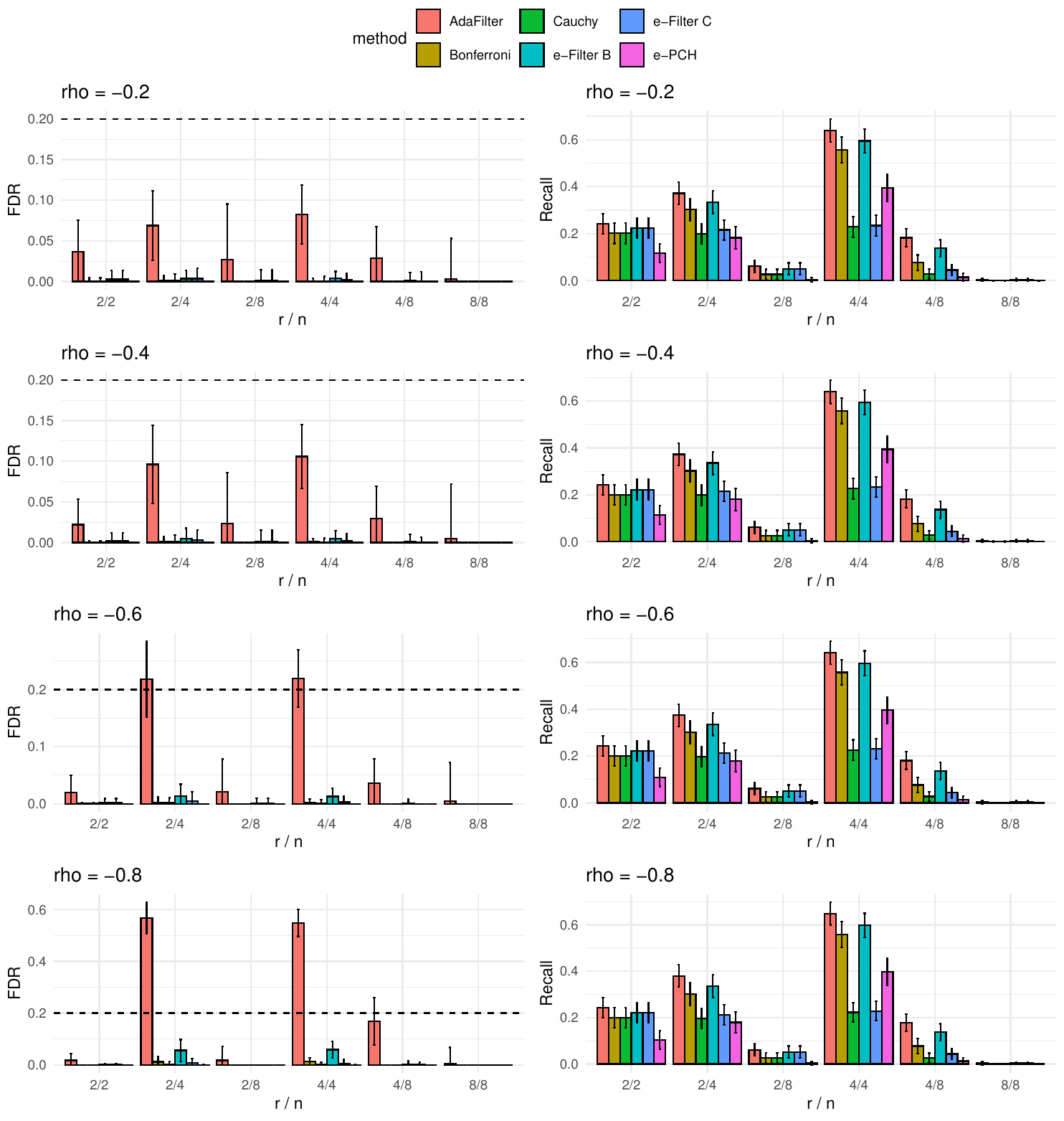}
\caption{Scenario 3: Average FDR and Recall of various methods across $B=500$ repetitions of the data generation scheme. Here $\alpha=0.2$. {The error bars represent one standard deviation above and below the average FDR (Recall) from $B$ repetitions.} Table \ref{tab:setting_2_fdr} provides the summary over the $r$ by $n$ combinations.}
\label{fig:scenario_2_rbyn}
\end{figure}
\begin{figure}[!h]
\centering
\includegraphics[width=1\linewidth]{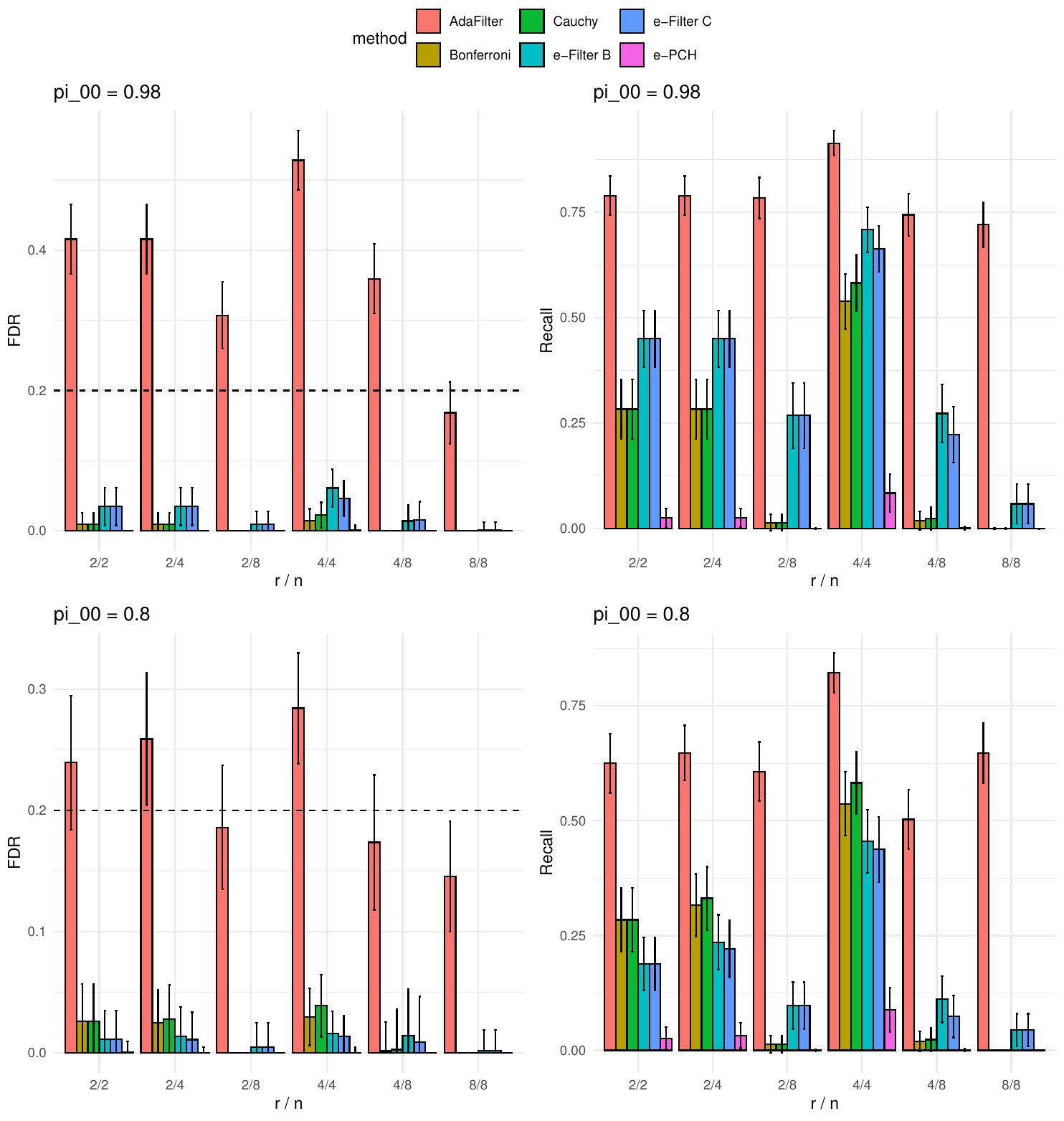}
\caption{Scenario 4: Average FDR and Recall of various methods across $B=500$ repetitions of the data generation scheme. Here $\alpha=0.2$. {The error bars represent one standard deviation above and below the average FDR (Recall) from $B$ repetitions.}  Table \ref{tab:scenario_4_fdr} provides the summary over the $r$ by $n$ combinations.}
\label{fig:scenario_4_fdr_rbyn}
\end{figure}
\begin{figure}[!h]
\centering
\includegraphics[width=1.05\linewidth]{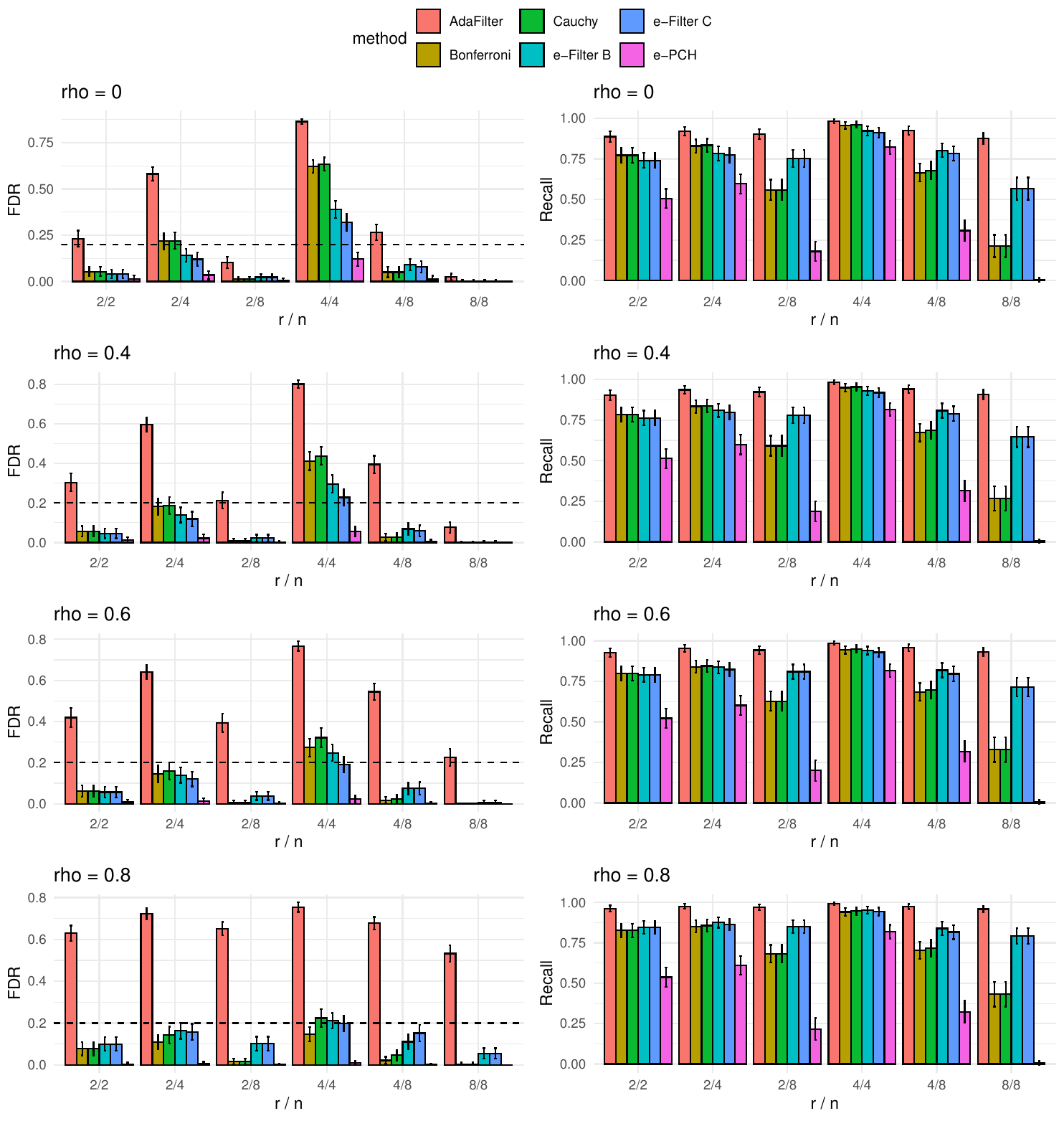}
\caption{Scenario 5: Average FDR and Recall of various methods across $B=500$ repetitions of the data generation scheme. Here $\alpha=0.2$. {The error bars represent one standard deviation above and below the average FDR (Recall) from $B$ repetitions.} Table \ref{tab:setting_5_fdr} provides the summary over the $r$ by $n$ combinations.}
\label{fig:scenario_5_rbyn}
\end{figure}
\end{document}